\documentclass[runningheads]{llncs}

\newif\ifarxivstyle
\arxivstyletrue

\spnewtheorem*{theorem*}{Theorem}{\normalfont\bfseries}{\itshape}
\spnewtheorem*{corollary*}{Corollary}{\normalfont\bfseries}{\itshape}

\usepackage{thmtools}
\usepackage{thm-restate}
\usepackage{etoolbox}
\usepackage{amsmath,amsfonts,amssymb,amscd}
\usepackage{algorithm2e,algorithmic}
\usepackage{lastpage}
\usepackage{enumerate}
\usepackage{fancyhdr}
\usepackage{xcolor}
\usepackage{tcolorbox}
\usepackage{graphicx}
\usepackage{listings}
\usepackage[
    bookmarks=false,
    colorlinks=true,
    citecolor=orange,
    linkcolor=blue!55!black,
    urlcolor=blue!55!black,
    backref=page
]{hyperref}
\usepackage{dsfont}
\usepackage{tikz}
\usetikzlibrary{arrows.meta}
\usepackage{circuitikz}
\usepackage[braket, qm]{qcircuit}
\usepackage{mathtools}
\usepackage{stmaryrd}
\usepackage{cite}
\usepackage{array}
\usepackage{float}
\usepackage{multirow}
\usepackage{float}
\usepackage{enumitem}
\usepackage[normalem]{ulem}
\usepackage[
    operators,
    adversary,
    complexity,
    advantage,
    probability,
    lambda,
    mm,
    sets,
    asymptotics
]{cryptocode}
\usepackage[capitalize]{cleveref}

\usepackage{fontawesome5}
\definecolor{orcidlogocol}{rgb}{0.65, 0.807, 0.223}
\newcommand{\orcid}[1]{$\,$\href{https://orcid.org/#1}{\textcolor{orcidlogocol}{\faOrcid}}}

\def\submission{0}
\ifnum\submission=0
	\author{
        Ritam Bhaumik\inst{1}\orcid{0000-0002-2883-4870}\and
        Yu-Hsuan Huang\inst{2}\orcid{0009-0002-6090-6712}
    }\institute{
        Technology Innovation Institute (TII), United Arab Emirates
        \and
        Max Planck Institute for Security and Privacy (MPI-SP), Germany\\
        \email{ritam.bhaumik@tii.ae, yu-hsuan.huang@mpi-sp.org}
    }
\else
	\author{}
	\institute{}
\fi

\ifnum\submission=0
	\newcommand{\yuhsuan}[1]{\textcolor{cyan}{$\langle \langle$ YH: #1 $\rangle \rangle$}}
    \newcommand{\ritam}[1]{\textcolor{purple}{$\langle \langle$ RB: #1 $\rangle \rangle$}}
\else
	\newcommand{\yuhsuan}[1]{}
    \newcommand{\ritam}[1]{}
\fi

\spnewtheorem{claim}{Claim}{\itshape}{\rmfamily}
\crefname{claim}{Claim}{Claims}
\Crefname{claim}{Claim}{Claims}

\makeatletter
\renewenvironment{proof}[1][\proofname]{%
    \topsep 7\p@ \@plus2\p@ \@minus4\p@
    \@Begintheorem{#1}{\itshape}{\rmfamily}\ignorespaces
}{\@endtheorem}
\makeatother

\newcommand{\Ber}{\mathsf{Ber}}
\newcommand{\sgn}{{\mathsf{sgn}}}

\newcommand{\TWK}{{\sf TWK}}
\newcommand{\DM}{{\sf DM}}
\newcommand{\PRF}{{\sf PRF}}
\newcommand{\PRP}{{\sf PRP}}
\newcommand{\IC}{{\sf IC}}
\newcommand{\RP}{{\sf RP}}
\newcommand{\RO}{{\sf RO}}

\newcommand{\Sym}{{\sf Sym}}
\newcommand{\cpO}{{\sf cpO}}
\newcommand{\cIC}{{\sf cIC}}
\newcommand{\Dom}{{\sf Dom}}
\renewcommand{\Im}{{\sf Im}}
\newcommand{\lrp}[1]{\left({#1}\right)}

\newcommand{\Adv}{{\mathbf{Adv}}}

\newcommand{\idx}{\mathds{1}}

\usepackage{cleveref}

\definecolor{blizzardblue}{rgb}{0.67, 0.9, 0.93}
\definecolor{lavender}{rgb}{0.96, 0.73, 1.0}

\newcommand{\Hyb}{\mathsf{Hyb}}

\newcommand{\sqbkt}[1]{\left[{#1}\right]}

\newcommand{\SD}{{S\hspace{-0.3ex}D}}

\newcommand{\bool}{\ensuremath{{\{0,1\}}}}

\usepackage{cryptocode}

\makeatletter
\renewcommand{\gamechange}[2][gamechangecolor]{%
  \begingroup
  \setlength{\fboxsep}{0pt}%
  \ifmmode
    \mathpalette{\gamechange@math{#1}}{#2}%
  \else
    \colorbox{#1}{#2}%
  \fi
  \endgroup
}

\newcommand{\gamechange@math}[3]{%
  \colorbox{#1}{$\m@th#2#3$}%
}
\definecolor{gamechangecolor}{RGB}{255, 255, 0}
\makeatother

\providecommand{\Acal}{ {\cal A} }

\providecommand{\Bcal}{ {\cal B} }

\providecommand{\Ccal}{ {\cal C} }

\providecommand{\Dcal}{ {\cal D} }

\providecommand{\Ecal}{ {\cal E} }

\providecommand{\Jcal}{ {\cal J} }

\providecommand{\Kcal}{ {\cal K} }

\providecommand{\Ocal}{\ifarxivstyle%
\mathscr{O}%
\else%
{\cal O}%
\fi}

\providecommand{\Ucal}{ {\cal U} }

\providecommand{\Xcal}{ {\cal X} }

\providecommand{\Ycal}{ {\cal Y} }

\providecommand{\Zcal}{ {\cal Z} }

\providecommand{\Csf}{ {\sf C} }

\providecommand{\Dsf}{ {\sf D} }

\providecommand{\Fsf}{ {\sf F} }

\providecommand{\Gsf}{ {\sf G} }

\providecommand{\Ksf}{ {\sf K} }

\providecommand{\Msf}{ {\sf M} }

\providecommand{\Osf}{ {\sf O} }

\providecommand{\Psf}{ {\sf P} }

\providecommand{\Qsf}{ {\sf Q} }

\providecommand{\Rsf}{ {\sf R} }

\providecommand{\Xsf}{ {\sf X} }

\providecommand{\Ysf}{ {\sf Y} }

\providecommand{\Zsf}{ {\sf Z} }

\providecommand{\Ebb}{ {\mathbb E} }

\providecommand{\Ibb}{ {\mathbb I} }

\providecommand{\Qbb}{ {\mathbb Q} }

\providecommand{\Rbb}{ {\mathbb R} }

\providecommand{\Zbb}{ {\mathbb Z} }

\DeclareFontFamily{U}{mathb}{}
\DeclareFontShape{U}{mathb}{m}{n}{<-5.5> mathb5 <5.5-6.5> mathb6 
<6.5-7.5> mathb7 <7.5-8.5> mathb8 <8.5-9.5> mathb9 <9.5-11> mathb10 
<11-> mathb12}{}
\DeclareSymbolFont{mathb}{U}{mathb}{m}{n}
\DeclareFontSubstitution{U}{mathb}{m}{n}
\DeclareMathSymbol{\blackdiamond}{\mathbin}{mathb}{"0C}  
\DeclareFontFamily{U}{mathb}{\hyphenchar\font45}
\DeclareFontShape{U}{mathb}{m}{n}{
      <5> <6> <7> <8> <9> <10>
      <10.95> <12> <14.4> <17.28> <20.74> <24.88>
      mathb10
      }{}
\DeclareSymbolFont{mathb}{U}{mathb}{m}{n}
\DeclareMathSymbol{\sqbullet}{1}{mathb}{"0D}
\DeclareFontFamily{U}{mathb}{\hyphenchar\font45}
\DeclareFontShape{U}{mathb}{m}{n}{
      <5> <6> <7> <8> <9> <10> gen * mathb
      <10.95> mathb10 <12> <14.4> <17.28> <20.74> <24.88> mathb12
      }{}
\DeclareSymbolFont{mathb}{U}{mathb}{m}{n}
\DeclareFontSubstitution{U}{mathb}{m}{n}
\DeclareMathSymbol{\smalltriangleup}   {2}{mathb}{"98}
\DeclareMathSymbol{\smalltriangledown} {2}{mathb}{"99}
\DeclareMathSymbol{\blacktriangleup}   {2}{mathb}{"9C}
\DeclareMathOperator*{\expn}{\ensuremath{\Ebb}}
\def\stirling#1#2{\genfrac\{\}{0pt}{}{#1}{#2}}
\newcommand{\mpow}[2]{{#1}^{\diamond #2}}
\newcommand{\mfall}[2]{({#1})_{\diamond #2}}
\newcommand{\gfall}[3]{({#1}\,|\,{#2})_{#3}}
\newcommand{\rowsum}[2]{\abs{{#1}_{{#2},\star}}}
\newcommand{\colsum}[2]{\abs{{#1}_{\star,{#2}}}}
\newcommand{\SZ}{{\sf SZ}}
\newcommand{\PRMG}{{\sf PRMG}}

\newcounter{todoctr}
\ifnum\submission=0
	\newcommand{\todo}[2]{\refstepcounter{todoctr}%
	\begin{tcolorbox}[colback=red!5!white,colframe=red!70!black,title={\small\bfseries TODO~\thetodoctr: #1}]\small #2\end{tcolorbox}}
\else
	\newcommand{\todo}[2]{}
\fi

\AtEndEnvironment{proof}{\leavevmode\qed}

\ifnum\submission=0
    \ifarxivstyle
        \usepackage[a4paper,margin=1in]{geometry}
        \usepackage{newpxtext,newpxmath}
        \usepackage{cite}
    \else
        \usepackage{mathrsfs}
    \fi
\else
    \usepackage{mathrsfs}
\fi

\makeatletter
\let\NAT@parse\undefined
\makeatother
\renewcommand*{\backref}[1]{}
\renewcommand*{\backrefalt}[4]{%
    \ifcase #1 \textsc{(not cited.)}%
    \or \textsc{(cited on p.~#2.)}%
    \else \textsc{(cited on pp.~#2.)}%
    \fi
}

\title{Indistinguishability Lifting for Keyed Oracles,\\
Compressed Ideal Cipher,\,and More Applications}

\begin{document}

\maketitle
\begin{abstract}
    Cryptographic security proofs often involve an adversary interacting with a larger, {\em keyed oracle} that consists of (potentially exponentially) many independent instances of a smaller, {\em base oracle}. Examples include ideal ciphers, which provide access to an independent random permutation for each key, and random oracles, which provide an independent random bit string for each input. However, showing quantum indistinguishability between two such keyed oracles can be tricky, since a single query made by an adversary may involve a superposition that covers all instances of the base oracles simultaneously.

    In this paper, we establish a generic {\em indistinguishability lifting theorem} of the following form: if the two base oracles are indistinguishable under quantum queries, then their corresponding keyed oracles are too, up to an $O(q^2)$ multiplicative loss in distinguishing advantage, where $q$ is the number of queries made by the adversary. Our lifting theorem applies to both statistical and computational settings, and to oracles that are stateful as well. It is also optimal in that it matches the obvious Grover search attack for a certain (contrived) choice of oracles.

    As an immediate application, we extend Carolan's compressed permutation oracle to an efficiently implementable {\em compressed ideal cipher}, and use it to prove preimage resistance of the Davies-Meyer compression function in the quantum ideal cipher model. Thanks to our lifting theorem, the soundness of our compressed ideal cipher reduces to that of Carolan's oracle, and any further improvement on the latter would automatically carry over to the former.

    As our second application, we give a modular construction that doubles the message length of any quantum-secure strong pseudorandom permutation. Along the way, we show that an existing two-round tweakable Feistel construction is indistinguishable from a random permutation under quantum bidirectional queries. This is done via a dedicated polynomial-method argument, which may be of independent interest.
\end{abstract}

\section{Introduction}

\subsubsection{Indistinguishability in cryptography.}

Indistinguishability has played an indispensable role in modern cryptography. From Shannon's treatment of perfect secrecy \cite{Shannon49}, through the seminal works of Goldwasser and Micali \cite{STOC:GolMic82,GolMic84} on computational security of public-key encryption, to Yao's formulation of pseudorandomness \cite{FOCS:Yao82a}, a recurring theme is that information generated from two different sources may nevertheless be indistinguishable to an observer.

Along with the development of cryptography as a field, the use of indistinguishability has further been extended: not only as a property between two pieces of input data, but rather concerning the entire interaction an adversary can have with the cryptographic scheme \cite{GolGolMic86,LubRac88} at hand. This is often captured in the form of {\em indistinguishability between oracles}, wherein an adversary is allowed to make a bounded number of queries to one oracle $\Osf_0$ or another $\Osf_1$, and is tasked to distinguish which oracle it is interacting with.
This also serves as a useful tool for security analyses. For example, one may perform a sequence of replacements to the oracle at hand, each of which remains indistinguishable, and in the end arrive at the final oracle that is harmless to the considered cryptographic scheme.

\subsubsection{Indistinguishability between keyed oracles.}

In cryptography, it is very common to consider a {\em keyed oracle}, which consists of (potentially exponentially) many independent instances of a base oracle $\Osf_b$, and the adversary is allowed to decide which instance of the oracle $\Osf_b$ to interact with in each query, indicated via a key $k\in\Kcal$. For instance, in the random oracle model \cite{CCS:BelRog93}, a random oracle consists of a family of independently sampled bit strings, one string per key, while in the ideal cipher model \cite{C:BlaRogShr02}, an ideal cipher consists of a family of independently sampled random permutation oracles, one permutation per key.

Naturally, one would hope to argue indistinguishability between two such keyed oracles $\Osf_0^\Kcal$ and $\Osf_1^\Kcal$. 
For this to hold, the indistinguishability of the base oracles $\Osf_0$ and $\Osf_1$ is clearly necessary, and establishing it depends on the specific oracles under consideration. This leads to a generic question:
\begin{quote}
    \textit{Does indistinguishability of the base oracles $\Osf_0$ and $\Osf_1$ imply that of the keyed oracles $\Osf_0^\Kcal$ and $\Osf_1^\Kcal$? If so, at what security loss?}
\end{quote}

Classically, the answer is yes, and a standard hybrid argument settles the security loss: for each instance of the base oracle $\Osf_0$ that the adversary has actually queried, we may replace it by the other base oracle $\Osf_1$, at an additive security loss $\epsilon$. Having replaced all $q$ instances that are queried by the adversary, we have effectively replaced the entire keyed oracle $\Osf_0^\Kcal$ with $\Osf_1^\Kcal$, at a total loss summing up to $q\epsilon$. This is the best one can achieve, since a matching attack exists for a certain (contrived) choice of oracles.\footnote{Concretely, think of the oracle $\Osf_b$ that outputs a pre-selected random bit from $\Ber(b\cdot\epsilon)$. Then $\Osf_0$ and $\Osf_1$ are $\epsilon$ indistinguishable but $\Osf_0^\Kcal$ and $\Osf_1^\Kcal$ are $\Omega(q\epsilon)$ distinguishable.}

In the quantum case, however, the situation becomes significantly trickier. Indeed, a quantum adversary may be able to query the keyed oracle in a quantum superposition of keys that, in turn, touches all instances of the base oracle simultaneously, and so the (naive) classical hybrid argument does not work in general. On the other hand, replacing the entire family of instances one by one often causes a security loss that explodes exponentially fast in the length of the keys, which is often unaffordable from the security perspective.

\subsubsection{Limitations of prior work.}
In the literature, there have been several ways to mitigate the aforementioned difficulties, but they all come with a price and limitations, which we discuss below.

The first limitation lies in the obtained distinguishing bound.
In the computational setting, Zhandry's small-range distribution technique \cite{FOCS:Zhandry12} (essentially) provides a way to simulate the keyed oracle $\Osf_b^\Kcal$ using another smaller keyed oracle $\Osf_b^{[L]}$ (for {\em stateless} oracles, which we discuss in the next paragraph) whose key space size is much smaller, of size $L$. An adversary $\Acal$ distinguishing this smaller keyed oracle $\Osf_b^{[L]}$ can then be transformed into an adversary $\Bcal$ distinguishing the underlying base oracles $\Osf_b$, via an $L$-step hybrid argument, which introduces a factor $L$ loss in the distinguishing advantage $\Adv(\Bcal)$. However, this small-range simulator is imperfect, with an additive security loss of order $O(q^3/L)$ for adversaries given $q$ quantum queries.
To make this additive simulation error negligible, the parameter $L$ then has to be set superpolynomially large, thereby introducing a superpolynomially large multiplicative loss in the final distinguishing bound. In particular, if one optimizes the overall bound via equating $L\cdot\Adv(\Bcal)$ and the simulation error $O(q^3/L)$, the achieved bound would be on the order of $q^{1.5}\sqrt{\Adv(\Bcal)}$, which is rather non-tight.
Further restricting to the statistical setting, where the base oracles $\Osf_0$ and $\Osf_1$ are sampled from distributions that are close in a statistical distance $\epsilon$, a dedicated tight query lower bound that matches Grover search can be achieved \cite{Bel19}. However, in terms of concrete advantage, prior work~\cite{AC:BDFLSZ11,C:SonYun17,EPRINT:BacMaj26} only reaches bounds of the form $O(q^{c}\sqrt{\epsilon})$ (for $c\geq 1$). That is, they all lose a square root in $\epsilon$.


Beyond quantitative loss, the second limitation lies in scope. All existing lifting results 
concern only {\em stateless} oracles. Namely, the base oracle $\Osf_b$ is a function (or unitary \cite{TCC:AGQY22}) that is sampled up front according to a certain distribution.
Many useful oracles are {\em not} of this kind. For instance, Zhandry's compressed oracle~\cite{C:Zhandry19} and the analogous compressed permutation oracles of
Unruh~\cite{AC:Unruh23} and Carolan~\cite{STOC:Car26} are stateful simulators whose quantum memory evolves with
every query. Moreover, their usefulness crucially relies on such memory, since a reduction may use it to embed/observe useful information to/from the adversary interacting with the oracle.
Extending, say, a compressed permutation oracle to a compressed \emph{ideal cipher}\,---\,which would be useful for analyzing permutation-based constructions\,---\,means lifting (quantum) indistinguishability of {\em stateful} base oracles $\Osf_0$ and $\Osf_1$, to the corresponding keyed oracles $\Osf_0^\Kcal$ and $\Osf_1^\Kcal$.
To the best of our knowledge, no prior work has established a generic lifting result for this stateful setting.

\subsection{Our Contributions}\label{sec:intro_contributions}

In this paper, we answer the above generic question affirmatively, to the best possible extent. Namely, we show that if (possibly stateful) oracles $\Osf_0$ and $\Osf_1$ are indistinguishable in $q$ quantum queries, then so are their corresponding keyed oracles $\Osf_0^\Kcal$ and $\Osf_1^\Kcal$, up to a factor $8q^2$ in the distinguishing bound.

\begin{theorem*}[Quantum indistinguishability lifting, informal]
    Every algorithm distinguishing $\Osf_0^\Kcal$ from $\Osf_1^\Kcal$ with $q$ quantum queries and advantage $\epsilon$ can be transformed into a similarly efficient algorithm distinguishing $\Osf_0$ from $\Osf_1$ with $q$ quantum queries and advantage at least $\epsilon / (8q^2)$.
\end{theorem*}

Our lifting theorem holds both statistically and computationally, and side-steps all the above limitations of prior works: One can immediately see that, contrary to prior works, our achieved indistinguishability bound does not have any square-root loss. Not only that, our bound is actually {\em optimal}: for a certain (contrived) choice of oracles $\Osf_0$ and $\Osf_1$ (see \cref{sec:blackbox_optimality}), the standard Grover search would have achieved a matching distinguishing advantage. Moreover, our indistinguishability lifting theorem holds for {\em stateful} oracles, making it applicable to more useful oracles, such as the (family of) compressed permutation oracles proposed by Unruh~\cite{AC:Unruh23} and Carolan~\cite{STOC:Car26}.

\subsubsection{Compressed ideal cipher.} As the first and immediate application of our indistinguishability lifting theorem, we show that existing compressed permutation oracle techniques in the literature \cite{AC:Unruh23,STOC:Car26,EPRINT:HHMRS26,FLMNW26} can be extended to the ideal cipher model. 
We do so via extending the particular oracle $\cpO$ considered in \cite{STOC:Car26}, to its corresponding keyed oracle $\cpO^\Kcal$. Thanks to our lifting theorem, we can readily lift the indistinguishability of $\cpO$ and a random permutation $P$, to that of $\cpO^\Kcal$ and the corresponding ideal cipher. Slightly less obvious though, is whether the new, keyed compressed permutation oracle admits an efficient implementation $\cIC$, which we address using a standard isometry-based argumentation.

In terms of concrete bounds, on a message space of $n$-bit strings, the oracle $\cpO$ is indistinguishable from a random permutation $P$ by \cite{STOC:Car26}, under a (single-instance) bound of order $O\lrp{{q^3}/{2^{n/4}}}$. Very recently, two works \cite{EPRINT:Kim26,ARXIV:CarMaj26} that are independent and concurrent to ours (see \cref{sec:concurrent}) further improve this bound to $O\lrp{{q}/{2^{n/2}}}$. This would in turn give us a bound for the distinguishing advantage between $\cIC$ and the corresponding ideal cipher $P^\Kcal$, under a bound of order $O\lrp{{q^3}/{2^{n/2}}}$. More generally, let $\Delta_\cpO(n,q)$ be the maximum distinguishing advantage between $\cpO$ and $P$ given $q$ quantum queries, our obtained bound for $\cIC$ and $P^\Kcal$ is $8q^2\cdot \Delta_\cpO(n,q)$ as spelled out below.

\begin{theorem*}[Compressed ideal cipher, informal]
    There exists an efficiently implementable compressed ideal cipher $\cIC$ such that every oracle algorithm making at most $q$ quantum bidirectional queries has distinguishing advantage at most
    $$
    8q^2\cdot \Delta_\cpO(n,q)
    $$
    between $\cIC$ and the ideal cipher $P^\Kcal$.
\end{theorem*}

Given the fast development toward improving single-instance bounds, we would like to highlight the usefulness of our indistinguishability lifting theorem. Indeed, our lifting theorem works generically, irrespective of the specifics of the underlying oracles. Therefore, should there be any further improvement to the underlying compressed permutation bound, we would immediately obtain a corresponding improved bound for the compressed ideal cipher, without having to spend additional efforts redoing the proof from scratch.

As an application of our compressed ideal cipher, we prove preimage resistance of the Davies-Meyer compression function in the quantum ideal cipher model. We do so via mimicking existing proof techniques in the quantum random oracle model \cite{EC:CFHL21}, bounding query-by-query the magnitude of the compressed ideal cipher database containing such a preimage.

\subsubsection{Domain extension of quantum PRPs.}
The second application of our lifting theorem is domain extension for quantum-secure strong pseudorandom permutations (PRPs).\footnote{Prior to our work, Bhaumik et al.\ also \cite{CiC:BCFMN24} studied domain extension of (forward-only) quantum PRP, which requires an additional random-access assumption.} We show that every quantum-secure strong PRP $\PRP$ on $n$-bit message space can be extended to another one that operates on $2n$-bit message space. This is done via a two-step, modular construction, as illustrated in \cref{fig:domain_extension_overview}.

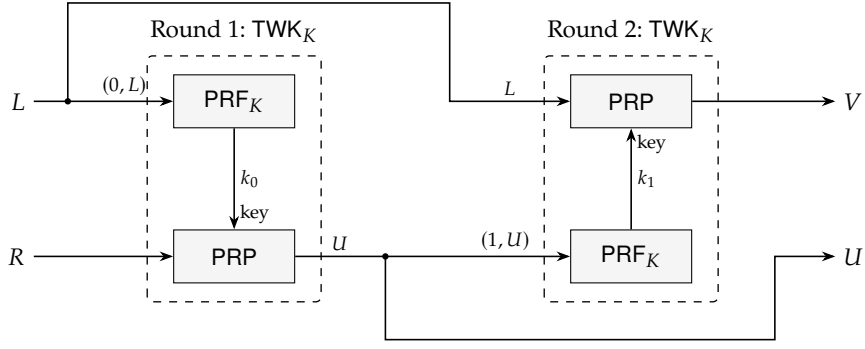
\begin{figure}[H]
    \centering
    \begin{tikzpicture}[
        x=1cm,y=1cm,
        font=\small,
        >={Stealth[length=1.8mm,width=1.3mm]},
        wire/.style={->,line width=.6pt},
        block/.style={draw,fill=black!4,minimum width=1.6cm,
            minimum height=.7cm,inner sep=3pt},
        roundbox/.style={draw,dashed,rounded corners=2pt,line width=.5pt},
        portlabel/.style={font=\scriptsize,inner sep=2pt}
    ]
        \draw[roundbox] (1.5,-.6) rectangle (3.8,2.65);
        \draw[roundbox] (6.75,-.6) rectangle (9.05,2.65);
        \node[inner sep=3pt] at (2.65,2.98)
            {Round 1: $\TWK_{K}$};
        \node[inner sep=3pt] at (7.9,2.98)
            {Round 2: $\TWK_{K}$};

        \node[block] (prf0) at (2.65,2.05) {$\PRF_{K}$};
        \node[block] (prp0) at (2.65,0) {$\PRP$};
        \node[block] (prp1) at (7.9,2.05) {$\PRP$};
        \node[block] (prf1) at (7.9,0) {$\PRF_{K}$};

        \draw[wire] (prf0.south) -- node[portlabel,right] {$k_0$} (prp0.north);
        \draw[wire] (prf1.north) -- node[portlabel,right] {$k_1$} (prp1.south);
        \node[portlabel,anchor=south west] at (prp0.north) {key};
        \node[portlabel,anchor=north west] at (prp1.south) {key};

        \node[anchor=east] at (0,2.05) {$L$};
        \node[anchor=east] at (0,0) {$R$};
        \draw[wire] (0,2.05) -- node[portlabel,above,pos=.65] {$(0,L)$} (prf0.west);
        \draw[wire] (0,0) -- (prp0.west);

        \fill (.45,2.05) circle (1.2pt);
        \draw[wire] (.45,2.05) -- (.45,3.35) -- (5.5,3.35)
            -- (5.5,2.05) -- node[portlabel,above] {$L$} (prp1.west);

        \coordinate (u) at (4.65,0);
        \draw[line width=.6pt] (prp0.east) -- node[portlabel,above] {$U$} (u);
        \draw[wire] (u) -- node[portlabel,above,pos=.65] {$(1,U)$} (prf1.west);
        \fill (u) circle (1.2pt);
        \draw[wire] (u) -- (4.65,-1.1) -- (9.8,-1.1) -- (9.8,0) -- (10.6,0);
        \draw[wire] (prp1.east) -- (10.6,2.05);
        \node[anchor=west] at (10.6,2.05) {$V$};
        \node[anchor=west] at (10.6,0) {$U$};
    \end{tikzpicture}
    \caption{Domain extension of a PRP from $n$-bit to $2n$-bit messages.}
    \label{fig:domain_extension_overview}
\end{figure}

In the first step, we construct a tweakable cipher $\TWK$ from $\PRP$, via using a pseudorandom function (PRF) $\PRF$ that, given any tweak $t$ as input, acts as a key-derivation procedure for $\PRP$:
$$
\TWK_{K}(t, x) := \PRP_{\PRF_K(t)}(x)\;.
$$
Here, we note that the required $\PRF$ can easily be obtained from $\PRP$ using the PRP-PRF switching lemma \cite{Zhandry15}.
To show that this generic construction of tweakable cipher is indeed quantum-secure, i.e.\ indistinguishable from an ideal cipher, we first replace queries to $\PRF_K(t)$ with queries to a random oracle $\RO(t)$. Then, the resulting oracle
$\TWK'_{K}(t,x) := \PRP_{\RO(t)}(x)$
can be seen as a keyed oracle $\PRP^\Kcal$ of $\PRP$ with an independent PRP key sampled for each tweak. Thanks to the fact that our indistinguishability lifting theorem also applies in the computational setting, we immediately conclude that it is indistinguishable from an ideal cipher. Prior to our work, this generic construction of tweakable cipher has been studied \cite{DCC:LLMM18} in the classical setting, but the quantum security remained untreated.

In the second step, we use the two-round variant of the domain extender proposed in \cite{TCC:CDMS10}, whose classical strong pseudorandomness is known, to construct the intended $2n$-bit PRP from the tweakable cipher $\TWK$. We further extend this strong pseudorandomness to the quantum setting using a tailored polynomial-method analysis, which may be of independent interest.

\subsection{Concurrent Work}\label{sec:concurrent}

Two very recent works, by Kim~\cite{EPRINT:Kim26} and by Carolan and Majenz~\cite{ARXIV:CarMaj26}, which are independent and concurrent to ours, revisit the soundness of Carolan's compressed permutation oracle $\cpO$~\cite{STOC:Car26}. Both improve its single-instance bound from $O\lrp{{q^3}/{2^{n/4}}}$ to $O\lrp{{q}/{2^{n/2}}}$ using different routes. Kim gives a concrete bound of $4q/2^{n/2}$ for $q\leq 2^n/4$, which is also formalized in Lean. As applications, Kim obtains tight query complexities for preimage and collision search in the sponge, and for collision search in the keyless Davies-Meyer construction. Carolan and Majenz construct the compression isometry directly from a POVM on the naive purification of a random permutation, and use representation theory of the symmetric group to control its normalization. As applications, they give explicit-constant bounds for the sponge, as well as preimage and collision bounds for the Davies-Meyer construction.

As discussed in \cref{sec:intro_contributions}, plugging these improved bounds into our lifting theorem immediately gives a bound of order $O\lrp{{q^3}/{2^{n/2}}}$ between our compressed ideal cipher $\cIC$ and the ideal cipher. Notably, Carolan and Majenz also treat the ideal cipher model directly: they give each key an independent compressed database, which is the same structure as our $\cIC$, and observe that a query under a key $k$ acts only on the $k$-th database. Their one-query analysis then carries over without any dependence on the number of keys, giving a bound of order $O\lrp{{q}/{2^{n/2}}}$ for the ideal cipher as well. For the specific oracle $\cpO$, their direct analysis is therefore tighter than our generic route by a factor of $q^2$. Likewise, their Davies-Meyer preimage bound of order $O(q^2/2^n)$ is tighter than ours in \cref{thm:DM_preimage}, whose dominant term comes from the lifting. We view the two approaches as complementary. Their extension relies on the structure of their soundness proof for $\cpO$, whereas our lifting theorem treats the base oracles as black boxes, and so applies to arbitrary (possibly stateful) base oracles, including other compressed permutation oracles whose soundness proofs may not have this structure.

\subsection{Other Related Work}

\subsubsection{Quantum security of symmetric cryptography.}
Our work belongs to the broader study of the quantum security of symmetric constructions. Kaplan et al.~\cite{C:KLLN16} showed that adversaries with superposition access to the keyed construction can break several constructions that are secure classically. On the positive side, quantum security has been established for message authentication~\cite{EC:BonZha13,EC:AMRS20}, for HMAC and NMAC~\cite{C:HosIwa21}, and for authenticated encryption from tweakable block ciphers~\cite{AC:BBCLNS21}, as well as for the sponge construction in several settings~\cite{EPRINT:CBHSU17,C:CarPor24,EPRINT:CarPorZha24,FOCS:ACMT25}. For the FX key-length extension~\cite{TCC:JaeSonTes21}, Even-Mansour cipher~\cite{EC:ABKM22}, and the Key-Alternating Feistel cipher~\cite{AC:BBCJJR25}, (various notions of) security has been established against adversaries with classical access to the keyed construction and quantum access to the underlying primitive. Closer to our applications, Zhandry~\cite{EPRINT:Zhandry16b} constructed quantum-secure PRPs from quantum-secure PRFs and classical format-preserving encryption, obtaining the first quantum-secure PRP from one-way functions, Hosoyamada and Iwata~\cite{AC:HosIwa19} attempted to show that the 4-round Luby-Rackoff construction is a quantum-secure PRP, but a gap in their proof was subsequently identified by Bhaumik et al.~\cite{AC:BCEJ24}; eventually, Carolan~\cite{STOC:Car26} showed that the 7-round Feistel construction is a strong quantum-secure PRP. The quantum security of the Davies-Meyer construction was studied by Hosoyamada and Yasuda~\cite{AC:HosYas18}. For the analysis of quantum queries to random permutations, we also mention the permutation superposition oracles of~\cite{STOC:MajMalWal25} and the lower bounds for two-sided permutation inversion of~\cite{CiC:ABPS24b}.

\paragraph{Quantum security of PRG with advice.}
An active line of work \cite{AC:HhaXagYam19,FOCS:CGLQ20,EC:Liu23,EPRINT:HuaLia26} has established increasingly tighter bounds on the PRG security of random oracles, against quantum adversaries given $L$ (qu)bits of oracle-dependent advice. As noted in \cite{FOCS:CGLQ20}, in the classical advice setting \cite{AC:HhaXagYam19} obtained a bound of order $\tilde{O}(\sqrt[6]{{Lq^4}/{N} + {q^4}/{N}})$, where $q$ is the number of (quantum) random oracle queries, and $N$ is the size of the key space. This is then improved to $\tilde{O}({\sqrt[3]{L q/N + q^2/N}})$ by \cite{FOCS:CGLQ20}, the authors of which further extended the results to the quantum advice setting with a bound $\tilde{O}({\sqrt[19]{L^5 q / N + L^4 q^2/N}})$. The subsequent work \cite{EC:Liu23} improves the quantum advice bound to $O({\sqrt[3]{Lq/N} + \sqrt{q^2/N}})$. Huang~and~Liao \cite{EPRINT:HuaLia26} further improve it to $O({\sqrt[3]{Lq/N}+ q^2/N})$, replacing the square-root term with $O(q^2/N)$. They also establish a multi-challenge bound $O({\sqrt[3]{nL(q+n)/N}+ nq^2/N})$, where $n$ is the number of PRG challenges given to the adversary. Notably, the polynomial method is also used in \cite{EPRINT:HuaLia26}, though in a different way, to obtain a tighter PRG security bound.

\paragraph{Quantum security of sum of permutations.}
Two recent papers by Bhaumik, Guo, Guo, and Jha~\cite{ARXIV:BGGJ26} and Choi, Hhan, and Jang~\cite{EPRINT:ChoHhaJan26} study the quantum indistinguishability of the sum (or XOR) of independent random permutations from a random function, a classical way to convert permutations into functions beyond the birthday bound. Choi, Hhan, and Jang show that the XOR of $r\geq 2$ random permutations over $n$-bit strings is quantum-secure beyond the quantum birthday bound, while Bhaumik et al.\ treat sums over arbitrary finite abelian groups, together with linear postprocessing such as truncation and Dinur's LXoP construction. Both works use Fourier-analytic variants of the polynomial method, in which the advantage of a $q$-query quantum distinguisher is controlled by components of degree at most $2q$, determined by $2q$ input-output pairs of the permutations. This is close in spirit to our analysis of the two-round construction in \cref{sec:tweak2ext}, where we also write the output probability of a distinguisher as a polynomial of degree at most $2q$ in the input-output indicators of a random permutation, although the constructions and the resulting arguments are different. Their results also concern the first step of our domain extension, where we obtain $\PRF$ from $\PRP$ via the PRP-PRF switching lemma; a beyond-birthday conversion such as the XOR of permutations could replace this step, although the overall bound of our domain extension also contains an $O(q^4/2^n)$ term from the second step, which such a replacement would not affect.

\subsubsection{Other notions of lifting.}
Our lifting theorem can be regarded as a decision counterpart of existing lifting theorems for search games, which lift security under classical oracle access to security under quantum oracle access.
This line of work builds on top of the measure-and-reprogram technique \cite{C:DFMS19,C:DonFehMaj20}, originally developed for applications such as Fiat-Shamir transformations. Subsequent work by Yamakawa and Zhandry~\cite{EC:YamZha21} shows that it is more broadly applicable for search games in the quantum random oracle model. This family of lifting theorems is further expanded by Cojocaru et al.~\cite{C:CHLYY25} to the quantum bidirectional random permutation model, as well as the quantum ideal cipher model.
In contrast, our lifting theorem concerns {\em decision} games, with quantum oracle access both before and after the lifting. Interestingly, our lifting theorem can also recover a weaker form of the aforementioned {\em search} lifting, which we will discuss in \cref{sec:search_lifting}, whereas the existing search lifting theorems do {\em not} directly imply our result.

\subsubsection{Classical proofs with ideal permutations and ideal ciphers.}
Classically, ideal permutations and ideal ciphers underlie many security proofs, including those of the Even-Mansour cipher~\cite{AC:EveMan91,JC:EveMan97} and key-alternating ciphers~\cite{EC:BKLSST12,EC:CheSte14}, and the equivalence between the random oracle and ideal cipher models~\cite{C:CorPatSeu08,STOC:HolKunTes11}. Tweakable block ciphers were introduced by Liskov, Rivest, and Wagner~\cite{C:LisRivWag02}, and Minematsu~\cite{FSE:Minematsu09} built $2n$-bit block ciphers from $n$-bit tweakable block ciphers. In classical proofs in the ideal cipher model, the permutation associated with each key is typically sampled lazily when the key is first queried, so that only the queried keys enter the analysis. This is exactly the step that fails against adversaries querying the key in superposition, and that our lifting theorem replaces.

\section{Preliminaries}

\subsection{Basic Notation}
For every integer $n\in\Zbb_{>0}$ we write $\sqbkt{n} := \set{1,\dots,n}$, with the convention that $[0] := \emptyset$ is the empty set. For a predicate or event $P$, we write $\idx_P$ for the indicator function that evaluates to $1$ if it is satisfied, and $0$ otherwise. When $P$ is described in a long expression, we may also write $\idx\sqbkt{P}$ to make the formula more neat.

We use the falling Pochhammer (or falling-factorial) notation
\[
    (a)_k := a(a-1)\cdots(a-k+1) \quad (k\geq 1),
    \qquad (a)_0:=1.
\]
We write $\stirling{k}{j}$ for the Stirling number of the second kind, namely the
number of partitions of a $k$-element set into $j$ nonempty, unlabeled blocks.
Equivalently, these numbers are the coefficients in the change of basis
\begin{equation}\label{eq:stirling}
    x^k=\sum_{j=0}^k \stirling{k}{j}(x)_j.
\end{equation}

\subsection{Oracle Notation}\label{sec:oracle_prelim}
For an algorithm/procedure $\Acal$ that produces binary output, we denote by $\Pr\sqbkt{\Acal}$ the probability that $\Acal$ outputs $1$, and likewise we denote by $\SD(\Acal,\Bcal)$ the statistical distance of the two output distributions of the two algorithms, respectively.

Throughout this paper, we write $\Acal^O$ for an algorithm given \emph{controlled} oracle access to $O$, meaning that in each query it can decide on the fly whether to disable the query. In case $O$ is a permutation or an invertible cipher, we take it as understood that $\Acal$ can make bidirectional queries as well. This is then further emphasized with the notation $O^{\leftrightarrow}$.

Moreover, for any PRF construction $\PRF$, we define
\[
\Adv^{\RO}_{\PRF}(\Acal) := \SD\lrp{\Acal^{\PRF_K},\Acal^{\RO}}
\]
to be the distinguishing advantage from a random oracle $\RO$ over the randomness of key $K$. Likewise, we denote by $\Adv^{\RP^{\leftrightarrow}}_{\PRP}(\Acal)$ and $\Adv^{\IC^{\leftrightarrow}}_{\TWK}(\Acal)$ the advantage as a PRP, and advantage as a tweakable cipher, given quantum bidirectional queries.

\subsubsection{Stateful oracles.} A (stateful) oracle $\Osf$ consists of a memory register $\Msf$, an initial state $\rho$ on $\Msf$, and a unitary $U_\Osf$ acting on a query register $\Qsf$ and on $\Msf$. This is without loss of generality: any randomness or intermediate measurement of the oracle can be purified into $\rho$ and $\Msf$. A controlled query to $\Osf$ applies $\ket{0}\bra{0}_\Csf\otimes\Ibb+\ket{1}\bra{1}_\Csf\otimes U_\Osf$ for a control qubit $\Csf$ held by the algorithm; in particular, a disabled query leaves $\Msf$ untouched. For a bidirectional oracle, $\Qsf$ contains a direction bit $d$ and $U_\Osf=\sum_{d}\ket{d}\bra{d}\otimes U_\Osf^d$, so the direction may be in superposition. An algorithm making $q$ queries applies its own unitaries on its registers, interleaved with $q$ (controlled) queries, and finally measures its output.

For a finite key set $\Kcal$, the keyed oracle $\Osf^\Kcal$ has memory $\bigotimes_{k\in\Kcal}\Msf_k$, where each $\Msf_k$ is a copy of $\Msf$, initial state $\rho^{\otimes\Kcal}$, and query unitary
\[
    U_{\Osf^\Kcal}:=\sum_{k\in\Kcal}\ket{k}\bra{k}_{\Ksf}\otimes \lrp{U_\Osf}_{\Qsf\Msf_k}\;,
\]
where $\Ksf$ is the key register and $(U_\Osf)_{\Qsf\Msf_k}$ acts as $U_\Osf$ on $\Qsf\Msf_k$ and as the identity on $\Msf_{k'}$ for $k'\neq k$. In other words, the instances start in independent copies of $\rho$, and each query is routed to the instance named by the key.

\subsection{Polynomial Inequalities}
We recall here Markov brothers' inequality, which is used multiple times in this paper. Let $f(x)$ be a (real or complex valued) polynomial, and let $[a,b] \subseteq \Rbb$ be a real interval. Then the Markov brothers' inequality would tell us the following
\begin{equation}\label{eq:Markov_brother}
    \abs{f'(x)} 
    \leq \frac{2(\deg f)^2}{b-a}\cdot \max_{t\in [a,b]} \abs{f(t)}\quad \forall x\in [a,b]\;.
\end{equation}
When we are interested in the interior points $x\in(a,b)$ away from the boundary, then Bernstein's inequality would give us a linear bound in the degree of $f$:
\begin{equation}\label{eq:Bernstein}
    \abs{f'(x)} \leq \frac{\deg f}{\sqrt{(x-a)(b-x)}}\cdot \max_{t\in[a,b]}\abs{f(t)}\quad \forall x\in (a,b)\;.
\end{equation}

\subsubsection{A Concrete Variant.}
We recall a concrete variant of Markov brothers' inequality that does not rely on maximization of $t$.
Let $d \geq \deg f$, and for every $0\leq j\leq d$, let $x_j := \cos(j\pi/d)$, and
\[
L_j(x) := \prod_{i\in\set{0,\dots,d}\setminus\set{j}}\frac{x-x_i}{x_j-x_i}
\]
be the $j$th element of the corresponding Lagrange basis, i.e.\ $L_j(x_i) = \idx_{i=j}$ and $\deg L_j = d$. Then we immediately get
\begin{align*}
    f'(x)
    = \lrp{\sum_{j\in\set{0,\dots, d}} L_{j}(x) \cdot f(x_j)}'
    = \sum_{j\in\set{0,\dots, d}} L_{j}'(x) \cdot f(x_j)
\end{align*}
Observe that, from the Duffin-Schaeffer inequality \cite{ToAMS:DufSch41} we have that $\sum_{j}\abs{L_j'(x)} \leq d^2$ for every $x\in[-1,1]$, and so in this case we have
\begin{equation*}
    \abs{f'(x)} \leq d^2\abs{\expn_{X\gets\Dcal_x^{d}}\sqbkt{\sgn_X^x\cdot f(X)}}\qquad\forall x\in[-1,1]\;,
\end{equation*}
where $\sgn_X^x \in\set{\pm 1}$ is the sign\footnote{Whether the sign at zero is $1$ or $-1$ is irrelevant to our discussion.} of $L_j'(x)$ for the only $j$ such that $x_j = X$, and $\Dcal_x^{d}$ is the distribution over $\set{x_0,\dots, x_{d}}$ with
\[
\Pr_{X\gets \Dcal_x^{d}}\sqbkt{X = x_j} := \frac{\abs{L'_j(x)}}{\sum_{i}\abs{L_i'(x)}}\qquad\forall j \in \set{0,\dots,d}\;.
\]
For an $x\in[a,b]$ sitting in the general interval $[a,b]$, the similar bound as follows can be achieved via rescaling the interval to $[-1,1]$
\begin{equation}
    \abs{f'(x)} \leq \frac{2d^2}{b-a}\abs{\expn_{T\gets\Dcal_{x,[a,b]}^{d}}\sqbkt{\Delta\sgn_{T}^{x}\cdot f(T)}}\qquad\forall x\in[a,b]\;,
    \label{eq:concrete_Markov}
\end{equation}
where $\Delta\sgn^{x}_T := \sgn^{\Delta(x)}_{\Delta(T)}$ for $\Delta(x) := (2x-a-b)/(b-a)$, and the distribution $\Dcal_{x,[a,b]}^{d}$ is defined as follows
\[
\Pr_{T\gets\Dcal_{x,[a,b]}^{d}}\sqbkt{T = t} := \Pr_{X\gets\Dcal_{\Delta(x)}^{d}}\sqbkt{X = \Delta(t)}\;.
\]
\subsubsection{Lifting integer-point upper bounds to an interval.}
An immediate consequence of \eqref{eq:Markov_brother} is that, if the interval length is sufficiently long, specifically if $b-a \geq 4(\deg f)^2 > 0$, then
\begin{equation}
    \abs{f(x)} \leq 2 \max_{n\in[a,b]\cap\Zbb} \abs{f(n)}\quad \forall x\in [a,b]\;,
    \label{eq:polynomial_integers}
\end{equation}
meaning that one can lift an upper bound on integer points of $f$ to all continuous points, if it applies to sufficiently many points to begin with.

This can be seen by picking $x\in[a,b]$ that maximizes $\abs{f(x)}$, its closest integer point $n\in[a,b]$, and the following chain of inequalities
\begin{align*}
    \abs{f(x)}
    &\leq \abs{f(n)}
    +\max_{t\in[a,b]}\abs{f'(t)}
    \leq \abs{f(n)}
    +\frac{2(\deg f)^2}{b-a}\abs{f(x)}
    \leq \abs{f(n)}
    +(1/2)\cdot\abs{f(x)}
\end{align*}
where the first inequality follows from basic calculus and the fact that $\abs{n-x}\leq 1$, the second inequality follows from \eqref{eq:Markov_brother}, and the last follows from the assumption that $b-a \geq 4(\deg f)^2$.
Moving the rightmost term to the left immediately concludes the proof of \eqref{eq:polynomial_integers}.

\section{Quantum Indistinguishability Lifting}\label{sec:lifting}
In this section, we formally state and prove the main theorem of this paper. Throughout, let $\Osf_0$ and $\Osf_1$ be any two (possibly stateful) oracles, and let $\Osf_0^\Kcal$ and $\Osf_1^\Kcal$ be their corresponding keyed oracles, as have been defined in \cref{sec:oracle_prelim}.

\begin{restatable}[Quantum indistinguishability lifting]{theorem}{thmIndLift}\label{thm:ind_lift}
    Let $\Acal$ be a binary-output oracle algorithm given at most $q$ quantum queries to either $\Osf_0^\Kcal$ or $\Osf_1^\Kcal$. Then, there exists a binary-output oracle algorithm $\Bcal$ making at most $q$ quantum controlled queries to $\Osf_i$, $\Osf_0^\Kcal$, and $\Osf_1^\Kcal$ respectively, such that
    \begin{equation}\label{eq:ind_lift}
        \SD\lrp{\Acal^{\Osf_0^\Kcal},\Acal^{\Osf_1^\Kcal}}
        \leq 8q^2\cdot \SD\lrp{\Bcal^{\Osf_0,\Osf_0^\Kcal,\Osf_1^\Kcal},\Bcal^{\Osf_1,\Osf_0^\Kcal,\Osf_1^\Kcal}}\;.
    \end{equation}
    Moreover, $\Bcal$ has a runtime similar to that of $\Acal$.
\end{restatable}
\begin{remark}
    If indistinguishability of $\Osf_0$ and $\Osf_1$ holds statistically, or if the oracles $\Osf_0$ and $\Osf_1$ and their initial states $\rho_0$ and $\rho_1$ are efficiently implementable, then the queries made by $\Bcal$ to $\Osf_0^\Kcal$ and $\Osf_1^\Kcal$ can be internalized.
\end{remark}

\begin{proof}
Our reduction algorithm $\Bcal$ works as follows. First, it samples a parameter $p^*\gets[0,1]$ uniformly and another parameter $z^*\gets\Dcal^{2q}_{0,[0,1]}$ following the distribution $\Dcal_{x,[a,b]}^d$ specified in \eqref{eq:concrete_Markov} (with $(x,a,b,d) := (0,0,1,2q)$). It then runs a parametrized subprocedure $\Bcal_{p,z}$ with $(p,z) := (p^*,z^*)$, and flips the output bit of $\Bcal_{p,z}$ whenever the sign $\Delta\sgn_{z^*}^0\in\set{\pm 1}$ (as defined in \eqref{eq:concrete_Markov}) is negative. Finally, it outputs the potentially flipped bit.

As for $\Bcal_{p,z}$, it maintains two functions $H$ and $G$ mapping from $\Kcal\to\bool$, which are independently sampled according to Bernoulli distributions $\Ber(p)$ and $\Ber(z)$ at each entry respectively. Then it runs $\Acal$ while simulating the oracle queries made by $\Acal$ through an interface $\Ocal$: on input $\ket{k,x}$ (or its superposition), forward the $\ket{x}$ part to $\Osf_i$ if $G(k)$ is one, and otherwise forward it to $\Osf_b^\Kcal$ with $b := H(k)$. We summarize the behavior of $\Bcal$ and $\Bcal_{p,z}$ below.
\begin{figure}[H]
    \centering
    \scalebox{1}{
    \begin{minipage}[t]{\ifarxivstyle%
        0.33\linewidth%
        \else%
        0.4\linewidth%
        \fi}
        $\Bcal^{\Osf_i,\Osf_0^\Kcal,\Osf_1^\Kcal}$:\hfill\unless\ifarxivstyle\vspace{-1em}\fi
        \begin{algorithmic}[1]
            \STATE $(p^*,z^*)\gets \Ucal([0,1])\otimes\Dcal^{2q}_{0,[0,1]}$
            \STATE $s^* := \idx\sqbkt{\Delta\sgn^{0}_{z^*}=-1}$
            \STATE {\bf return} $s^*\oplus \lrp{b\gets\Bcal^{\Osf_i,\Osf_0^\Kcal,\Osf_1^\Kcal}_{p^*,z^*}}$
        \end{algorithmic}
    \end{minipage}
    \begin{minipage}[t]{\ifarxivstyle%
        0.2\linewidth%
        \else%
        0.26\linewidth%
        \fi}
        $\Bcal^{\Osf_i,\Osf_0^\Kcal,\Osf_1^\Kcal}_{p,z}$:\hfill\vspace{-1em}
        \begin{algorithmic}[1]
            \STATE {\bf for} $k\in\Kcal$:
            \STATE \quad $H(k) \gets \Ber(p)$
            \STATE \quad $G(k) \gets \Ber(z)$
            \STATE {\bf return} $b\gets\Acal^{\Ocal}$
        \end{algorithmic}
    \end{minipage}
    \begin{minipage}[t]{\ifarxivstyle%
        0.25\linewidth%
        \else%
        0.33\linewidth%
        \fi}
        $\Ocal(k,x)$:~\hfill
        \begin{algorithmic}[1]
            \STATE {\bf if} $G(k)$ {\bf return} $\Osf_i(x)$
            \STATE {\bf else if} $H(k)$
            \STATE \quad {\bf return} $\Osf_{1}^\Kcal(k,x)$
            \STATE {\bf else} {\bf return} $\Osf_{0}^\Kcal(k,x)$
        \end{algorithmic}
    \end{minipage}}
    \caption{The (behavior of the) reduction algorithm $\Bcal$ for indistinguishability lifting, which runs a parametrized subprocedure $\Bcal_{p,z}$ with $p$ and $z$ substituted with $p^*$ and $z^*$ sampled at the beginning, and sometimes flips the output dependent on $\Delta\sgn^0_{z^*}\in\set{\pm 1}$ defined in \eqref{eq:concrete_Markov}. Here, $\Ucal([0,1])$ denotes the uniform distribution over $[0,1]$, and $\Dcal^{2q}_{0,[0,1]}$ is as defined in \eqref{eq:concrete_Markov}. Note that, while we only write the oracle interface $\Ocal$ in classical pseudocode, quantum queries are also handled in the obvious way.}\label{fig:lifting_reduction}
\end{figure}
\begin{remark}
    While in \cref{fig:lifting_reduction} the reduction algorithm looks inefficient, it is actually not.
    This is because the functions $H$ and $G$ are just (biased) random oracles, where each entry is sampled independently from a Bernoulli distribution with a parameter sampled at the beginning.
    This can be efficiently implemented using existing techniques. We discuss this matter in more detail in \cref{sec:efficient_biased_RO_simulation}.
\end{remark}

Our proof is performed via a rather non-obvious use of the polynomial method. Let $f(p)$ and $g_p(z)$ be functions defined as follows:
\begin{align}\label{eq:def_fg}
    f(p) := \Pr\sqbkt{\Bcal_{p,0}^{\Osf_i,\Osf_0^\Kcal,\Osf_1^\Kcal}}\;,\quad
    g_p(z) := \Pr\sqbkt{\Bcal_{p,z}^{\Osf_1,\Osf_0^\Kcal,\Osf_1^\Kcal}} - \Pr\sqbkt{\Bcal_{p,z}^{\Osf_0,\Osf_0^\Kcal,\Osf_1^\Kcal}}\;,
\end{align}
where the choice of $i\in\bool$ in the definition of $f(p)$ is irrelevant (since $\Osf_i$ is never queried).\footnote{Note that $f(p)$ and $g_p(z)$ are polynomials in $p$ and $z$ respectively, and so they are differentiable.} From a procedural perspective, the absolute difference $\abs{f(1) - f(0)}$ coincides with $\Acal$'s distinguishing advantage between $\Osf_0^\Kcal$ and $\Osf_1^\Kcal$. Therefore, we have
\begin{equation}\label{eq:lifting_mean_value}
    \SD\lrp{\Acal^{\Osf_0^\Kcal},\Acal^{\Osf_1^\Kcal}}
    = \abs{f(1) - f(0)}
    = \abs{\int_{0}^1 f'(p) \,dp}
    = \abs{\expn_{p^*}\sqbkt{f'(p^*)}}\;.
\end{equation}
Via a careful analysis deferred to \cref{lem:derivative_connection}, there is a connection between the derivatives of $f$ and $g$, namely $f'(p) = g_{p}'(0)$. Furthermore, by \cref{lem:g_polynomial}, the query complexity of $\Acal$ implies that $g_{p}(z)$ is a polynomial in $z$ of degree at most $2q$, and hence so is the polynomial $\bar{g}(z):=\expn_{p^*}\sqbkt{g_{p^*}(z)}$. The derivative $g'_{p}(0)$ can be controlled (on average over $p\sim p^*$) via
\begin{align}
    &\abs{\expn_{p^*}\sqbkt{g'_{p^*}(0)}}
    = \abs{\frac{d}{dz}\expn_{p^*}\sqbkt{g_{p^*}(z)}|_{z=0}}\notag\\
    &= \abs{\bar{g}'(0)}
    \leq 8q^2 \cdot \abs{\expn_{z^*}\sqbkt{\Delta\sgn^0_{z^*}\cdot \bar{g}(z^*)}}
    = 8q^2 \cdot \abs{\expn_{p^*,z^*}\sqbkt{\Delta\sgn^0_{z^*}\cdot g_{p^*}(z^*)}}\notag\\
    &= 8q^2\cdot \abs{\Pr\sqbkt{s^*\oplus\lrp{b\gets\Bcal_{p^*,z^*}^{\Osf_1,\Osf_0^\Kcal,\Osf_1^\Kcal}}} - \Pr\sqbkt{s^*\oplus\lrp{b\gets\Bcal_{p^*,z^*}^{\Osf_0,\Osf_0^\Kcal,\Osf_1^\Kcal}}}}\notag\\
    &= 8q^2\cdot\SD\lrp{\Bcal^{\Osf_0,\Osf_0^\Kcal,\Osf_1^\Kcal},\Bcal^{\Osf_1,\Osf_0^\Kcal,\Osf_1^\Kcal}}\;,\label{eq:reduction_gp_bar}
\end{align}
where the only (non-equality) inequality is via \eqref{eq:concrete_Markov}, and in the second to last equality $s^* := \idx_{\Delta\sgn^{0}_{z^*}=-1}$. We now explain the equalities one by one.
The first equality is via $g_p(z)$ being a polynomial whose coefficients are continuous in $p$, so that expectation commutes with derivative. The second equality is via the definition of $\bar{g}(z)$, and so is the third to last equality. The second to last equality is via the fact that $\Pr\sqbkt{1\oplus b_1} - \Pr\sqbkt{1\oplus b_0} = -(\Pr\sqbkt{b_1} - \Pr\sqbkt{b_0})$ for any binary random variables $b_0$ and $b_1$. The last equality is via the definition of $\Bcal$.

Collecting the bounds so far, we conclude the proof with
\begin{align*}
    \SD\lrp{\Acal^{\Osf_0^\Kcal},\Acal^{\Osf_1^\Kcal}}
    &= \abs{\expn_{p^*}\sqbkt{f'(p^*)}}
    =\abs{\expn_{p^*}\sqbkt{g'_{p^*}(0)}}
    \leq 8q^2\cdot \SD\lrp{\Bcal^{\Osf_0,\Osf_0^\Kcal,\Osf_1^\Kcal},\Bcal^{\Osf_1,\Osf_0^\Kcal,\Osf_1^\Kcal}}
\end{align*}
where the first equality follows from \eqref{eq:lifting_mean_value}, the second equality follows from the fact that $f'(p) = g'_p(0)$, and the only (non-equality) inequality follows from \eqref{eq:reduction_gp_bar}.
\end{proof}

In the remainder, we state and prove \cref{lem:derivative_connection,lem:g_polynomial} that are used in the above proof.
\begin{lemma}\label{lem:g_polynomial}
    For every $0<p<1$, the function $g_p(z)$ defined in \eqref{eq:def_fg} is a polynomial in $z$ of degree at most $2q$.
\end{lemma}
\begin{proof}
    For every $G^\circ:\Kcal\to\bool$, let
    \begin{equation}\label{eq:def_h_proof_deg}
        h_{p,z,i}(G^\circ) := \Pr\sqbkt{\Bcal_{p,z}^{\Osf_i,\Osf_0^\Kcal,\Osf_1^\Kcal}\;\Big|\; G = G^\circ}
    \end{equation}
    so that $g_p(z) = \expn_{G}\sqbkt{h_{p,z,1}(G) - h_{p,z,0}(G)}$.
    A standard argument in the polynomial method (see e.g. \cite{FOCS:BBCMW98}) implies that the function $h_{p,z,i}(G)$ is a multilinear polynomial in $(G(k))_{k\in\Kcal}$ of degree at most $2q$. Indeed, the oracle $G$ is only referenced in $\Ocal$ while running $\Bcal_{p,z}$, and each of the $q$ queries to $\Ocal$ implements the unitary
    \[
    \ket{k,x,y} \mapsto G(k) \ket{k}\otimes \Osf_i\ket{x,y}
    + (1-G(k)) \Osf^\Kcal_{H(k)}\ket{k,x,y}\;,
    \]
    which would only increase the degree of the amplitudes (as a polynomial in $(G(k))_k$) by one.
    Rewrite $h_{p,z,i}(G)$ in Fourier expansion (i.e.\ in terms of monomials):
    \begin{align*}
        h_{p,z,i}(G) = \sum_{S\subseteq\Kcal:\abs{S}\leq 2q} \widehat{h}_{p,z,i}(S) \prod_{k\in S} G(k)\;,
    \end{align*}
    and then exploiting the independence of $(G(k))_{k\in\Kcal}$, we have that
    \begin{align*}
        \expn_G\sqbkt{h_{p,z,i}(G)}
        &= \expn_G\sqbkt{\sum_{S\subseteq\Kcal:\abs{S}\leq 2q} \widehat{h}_{p,z,i}(S) \prod_{k\in S} G(k)}\\
        &= \sum_{S\subseteq\Kcal:\abs{S}\leq 2q} \widehat{h}_{p,z,i}(S) \prod_{k\in S} \expn_G\sqbkt{G(k)}
        = \sum_{S\subseteq\Kcal:\abs{S}\leq 2q} \widehat{h}_{p,z,i}(S)\cdot z^{\abs{S}} \;,
    \end{align*}
    is a polynomial in $z$ of degree at most $2q$, and hence
    \begin{align*}
        \deg_z g_p(z)
        &= \deg_z \lrp{ \expn_{G}\sqbkt{h_{p,z,1}(G)} - \expn_G\sqbkt{h_{p,z,0}(G)}}
        \leq 2q\;.
    \end{align*}
\end{proof}

\begin{lemma}\label{lem:derivative_connection}
    Let $f(p)$ and $g_p(z)$ be defined as in \eqref{eq:def_fg}. Then, for every $0<p<1$, we have
    \[
    f'(p) = g_p'(0)\;.
    \]
\end{lemma}
\begin{proof}
    For every $S\subseteq \Kcal$ let
    \begin{align*}
        W_{i,p}(S) := \Pr\sqbkt{\Bcal_{p,z}^{\Osf_i,\Osf_0^\Kcal,\Osf_1^\Kcal}\;\Big|\; G(k) \equiv \idx_{k\in S}}\;,\quad \Delta_{p}(S) := W_{1,p}(S) - W_{0,p}(S)\;,
    \end{align*}
    where the choice of $z$ in the above definition is irrelevant (since $z$ is used to sample $G$, which is fixed by the conditioning on $S$). We can then rewrite $g_p(z)$ as follows:
    \begin{align*}
        g_p(z)
        &= \sum_{S\subseteq\Kcal} \Pr\sqbkt{G(k) \equiv \idx_{k\in S}}\cdot \Delta_{p}(S)
        = \sum_{S\subseteq\Kcal} z^{\abs{S}}(1-z)^{\abs{\Kcal}-\abs{S}}\cdot \Delta_{p}(S)\;.
    \end{align*}
    Differentiating both sides with respect to $z$ at $0$, we get
    \[
    g_p'(0)
    = \sum_{k\in\Kcal}\Delta_{p}(\set{k}) - \abs{\Kcal}\Delta_{p}(\emptyset)
    = \sum_{k\in\Kcal}\Delta_{p}(\set{k})\;,
    \]
    where in the second equality we exploit that $\Delta_{p}(\emptyset)$ vanishes.

    Toward calculating $f'(p)$, let
    \[
    V(S) := \Pr\sqbkt{\Bcal^{\Osf_i,\Osf_0^\Kcal,\Osf_1^\Kcal}_{p,0}\;\Big|\;
    H(k)\equiv\idx_{k\in S}}\;,
    \]
    where again the choice of $i$ and $p$ in the above definition is irrelevant (since $\Osf_i$ is never queried, and $p$ is used to sample $H$ which is fixed by $S$).

    Then, by the fact that the random variables $(H(k))_{k\in\Kcal}$ are independent, we get
    \[
    f(p)
    = \sum_{S\subseteq\Kcal} \Pr\sqbkt{H(k)\equiv\idx_{k\in S}}\cdot V(S)
    = \sum_{S\subseteq\Kcal} V(S)\cdot\prod_{k\in\Kcal}\Pr\sqbkt{H(k)=\idx_{k\in S}}\;.
    \]
    Observe that, taking the derivative of $\Pr\sqbkt{H(k) = \idx_{k\in S}}$ with respect to $p$ would give us
    \[
    \frac{d}{dp}\Pr\sqbkt{H(k)=\idx_{k\in S}} = 
    \frac{d}{dp}\lrp{p\cdot\idx_{k\in S} + (1-p)\cdot\idx_{k\not\in S}}
    = \idx_{k\in S} - \idx_{k\not\in S}\;.
    \]
    Therefore, taking the derivative of $f(p)$, we obtain
    \begin{align*}
        f'(p)
        &= \sum_{k\in\Kcal}\sum_{S\subseteq\Kcal} V(S)\cdot\lrp{\idx_{k\in S} - \idx_{k\not\in S}}\prod_{k'\in\Kcal\setminus\set{k}}\Pr\sqbkt{H(k')=\idx_{k'\in S}}\;,
    \end{align*}
    Further observing that
    \begin{align*}
        &\sum_{S\subseteq\Kcal} V(S)\cdot \idx_{k\in S} \prod_{k'\in\Kcal\setminus\set{k}}\Pr\sqbkt{H(k')=\idx_{k'\in S}} = W_{1,p}(\set{k})\\
        &\sum_{S\subseteq\Kcal} V(S)\cdot \idx_{k\not\in S} \prod_{k'\in\Kcal\setminus\set{k}}\Pr\sqbkt{H(k')=\idx_{k'\in S}} = W_{0,p}(\set{k})\;,
    \end{align*}
    we immediately obtain
    \[
    f'(p) = \sum_{k\in\Kcal}\lrp{W_{1,p}(\set{k}) - W_{0,p}(\set{k})} = \sum_{k\in\Kcal}\Delta_p(\set{k}) = g_p'(0)\;,
    \]
    which concludes the proof.
\end{proof}

\subsection{On Black-box Optimality}\label{sec:blackbox_optimality}
It is easy to see that our bound in \cref{thm:ind_lift} cannot be further improved, unless relying on certain specifics of the considered oracles. Indeed, consider an oracle $\Osf_p$ parametrized by $p\in [0,1]$, where every output $\Osf_p(X)$ is set to a bit $b$ that is sampled at the beginning, with probability $p$ being one:
\[
\Osf_p(X) := b \;\;\forall X\;,\quad  b\stackrel{\text{init}}{\gets}\Ber(p)\;.
\]
Then the distinguishing advantage between $\Osf_0$ and $\Osf_\epsilon$ is at most $\epsilon$, and a simple Grover search would distinguish the keyed oracles $\Osf_0^\Kcal$ and $\Osf_\epsilon^\Kcal$ with advantage up to $\Omega(q^2\epsilon)$.

However, that is not to say our bound is tight for all choices of oracles. 
For example, the distinguishing advantage between $\Osf_{(1\pm\epsilon)/2}$ is $\epsilon$, but the distinguishing advantage between $\Osf_{(1\pm\epsilon)/2}^\Kcal$ can be further bounded by
\begin{align*}
    \SD\lrp{\Acal^{\Osf_{(1-\epsilon)/2}^\Kcal},\Acal^{\Osf_{(1+\epsilon)/2}^\Kcal}}
    &= \abs{f((1+\epsilon)/2) - f((1-\epsilon)/2)}
    \leq \int_{(1-\epsilon)/2}^{(1+\epsilon)/2} \abs{f'(p)} \,dp
\end{align*}
where $f(p) := \Pr\sqbkt{\Acal^{\Osf_p^\Kcal}}\in\Rbb[p]$ is a polynomial in $p$ of degree at most $2q$. 
Via Bernstein's inequality (\cref{eq:Bernstein}), we obtain $\abs{f'(p)} \leq 2q /\sqrt{p(1-p)}\leq 8q/\sqrt{3}$ for every $p\in[1/4,3/4]$. This would then tell us that, whenever $\epsilon\leq 1/4$,
\[
\SD\lrp{\Acal^{\Osf_{(1-\epsilon)/2}^\Kcal},\Acal^{\Osf_{(1+\epsilon)/2}^\Kcal}}
\leq \int_{(1-\epsilon)/2}^{(1+\epsilon)/2} (8q/\sqrt{3}) \,dp 
\leq 8q\epsilon/\sqrt{3}\;.
\]

We summarize the above discussion below.
\begin{proposition}
    There exists a pair of efficiently implementable oracles $\Osf_0$ and $\Osf_1$ that is $\epsilon$-indistinguishable, no matter how many queries are made, but for some oracle algorithm $\Acal$ making at most $q$ quantum queries, and a sufficiently large key space $\Kcal$, we have that
    \[
    \SD\lrp{\Acal^{\Osf_0^\Kcal}, \Acal^{\Osf_1^\Kcal}} \geq \Omega(q^2\epsilon)\;.
    \]
    On the other hand, there also exists a pair of efficiently implementable oracles $\Osf_0'$ and $\Osf_1'$ that is $\epsilon$-distinguishable in one classical query, but for every oracle algorithm $\Acal$ making at most $q$ quantum queries, we have that
    \[
    \SD\lrp{\Acal^{\Osf_0'^\Kcal}, \Acal^{\Osf_1'^\Kcal}} \leq O(q\epsilon)\;.
    \]

\end{proposition}

\section{Compressed Ideal Cipher and an Application}\label{sec:cIC_application}

\subsection{Carolan's Compressed Permutation Oracle}
Before looking into the case of ideal cipher, we first recall Carolan's compressed permutation oracle \cite{STOC:Car26}. Essentially, it is a simulator of random permutation which enjoys useful properties for security proofs.

Consider a message space $\Xcal := \bool^{n}$ that consists of all $n$-bit strings, a random permutation $P$ is sampled uniformly from $\Sym(\bool^n)$ the set of all permutations over $\Xcal$, and allows an algorithm $\Acal$ to make queries in the form of the following unitary
\[
O_P: \ket{x,y}\mapsto \ket{x, y\oplus P(x)}\;,
\]
for every $x,y\in \Xcal$. Likewise, we also allow $\Acal$ to make backward queries, which are of the form $O_{P^{-1}}$.

The compressed permutation oracle, which we denote by $\cpO$, is a stateful simulation of $P$ that maintains the same bidirectional query access, but differs in the way each query is responded to. The internal state of $\cpO$ is maintained in a register $\Dsf$, which we call the \emph{database register}. The state of this register is supported by vectors of the form $\ket{D}$, where $D:\Xcal\to\Xcal\cup\set{\bot}$, called a \emph{database}, is a partially specified permutation over $\Xcal$, with the special symbol $\bot$ denoting that a specific function value has not been specified. Specifically, for every distinct $x,y\in\Xcal$ such that $D(x)$ and $D(y)$ are non-$\bot$, it must be that $D(x)\neq D(y)$. As a short-hand notation, it is standard to write $\abs{D}$ for the number of non-$\bot$ entries, $\Dom(D)$ for the set of $x$ such that $D(x)\neq\bot$, $\Im(D)$ for the set of such $D(x)$'s, and $(x,y)\in D$ for $D(x) = y\neq\bot$. 

Initially, the database register $D:=\ket{\bot}$ is set to the constant-$\bot$ database, and each forward query is implemented by the following unitary
\[
\cpO^{\to}:\ket{x,y,D} \mapsto \ket{x} \Csf_x\cdot\Psf_x\cdot \Csf_x\ket{y,D}\;.
\]
Here, for every $x\in \Xcal$ and database $D$ such that $D(x)=\bot$, one writes
\[
\ket{+_{x,D}} := \frac{1}{\sqrt{2^n-\abs{D}}}\sum_{y\not\in\Im(D)}\ket{D[x\mapsto y]}
\]
for the uniform superposition of all (injective) databases that agree with $D$ outside of $x$; the \emph{compression operator} $\Csf_x$ maps $\ket{D}\mapsto\ket{+_{x,D}}$ and $\ket{+_{x,D}}\mapsto\ket{D}$ and fixes the orthogonal complement, and the operator $\Psf_x$ works as follows:
\begin{align*}
    &\Psf_x:\ket{y,D} \mapsto\begin{cases}
        \ket{y\oplus D(x), D}\text{ if }D(x)\neq\bot\\
        \ket{y, D}\qquad\qquad\text{otherwise}\;.
    \end{cases}
\end{align*}
For backward queries an additional flipping operator $\Fsf:\ket{D}\mapsto\ket{D^{-1}}$ with $(x,y)\in D$ if and only if $(y,x)\in D^{-1}$ is conjugated to the $\cpO^\to$ operator. Namely, the following unitary is applied
\[
\cpO^{\gets}:\ket{x,y,D}\mapsto \ket{x}\otimes\Fsf\cdot\Csf_x\cdot\Psf_x\cdot\Csf_x\cdot\Fsf \ket{y,D}\;.
\]

Indeed, as has been shown in \cite{STOC:Car26}, and subsequently sharpened by \cite{EPRINT:Kim26,ARXIV:CarMaj26}, this simulator is indistinguishable from the random permutation. In the statement below, let
\[
\Delta_{\cpO}(n,q) := \sup_{\Acal}\SD\lrp{\Acal^{P},\Acal^{\cpO}}
\]
where the supremum is taken over all oracle algorithms $\Acal$ making at most $q$ quantum queries to either $\cpO$ or $P$ bidirectionally, and $n$ specifies the bit-length of elements in $\Xcal$.

\begin{lemma}[Bounds achieved in \cite{EPRINT:Kim26,ARXIV:CarMaj26}]\label{lem:Carolan}
    For $n,q\in\Zbb_{>0}$, we have
    \[
    \Delta_{\cpO}(n,q) \leq O\lrp{\frac{q}{2^{n/2}}}\;.
    \]
\end{lemma}

\subsection{Compressed Ideal Cipher}\label{sec:cIC_def}
Now, we are ready to address the case of ideal cipher. Consider $P^\Kcal(k,x) := P_k(x)$ an oracle where $P_k$ is sampled independently from $\Sym(\Xcal)$ for each $k$. Again, we allow an algorithm to query this oracle bidirectionally, where a forward query takes as input $(k,x)\in\Kcal\times\Xcal$ and then responds $P_k(x)$, while a backward query responds $P_k^{-1}(x)$.

We analogously define the \emph{compressed ideal cipher} oracle $\cIC$, which is a stateful simulation of $P^\Kcal$ that maintains bidirectional query access. The internal database register $\Dsf$ is now supported by state vectors of the form $\ket{D}$, where $D:\Kcal\times\Xcal\to\Xcal\cup\set{\bot}$ must be a valid database. Namely, for every $k\in\Kcal$ and distinct $x,x'\in\Xcal$ such that $D(k,x)$ and $D(k,x')$ are non-$\bot$, it must be that they are distinct, i.e.\ $D(k,x)\neq D(k,x')$. For notational convenience, we will denote by $D_k(x) := D(k,x)$ the part of the database associated to the key $k$.

Initially, the database register $\Dsf$ is set to the constant-$\bot$ database, and, upon each query, the compressed ideal cipher oracle $\cIC$ would apply the operators $\cpO^{\to}$ (resp.\ $\cpO^{\gets}$) on $D_k$. That is, the forward and backward queries are respectively implemented by the following unitary operators:
\begin{align*}
    &\cIC^{\to}:\ket{k,x,y,D} \mapsto\ket{k,x} \Csf_{k,x}\cdot\Psf_{k,x}\cdot \Csf_{k,x}\ket{y,D}\\
    &\cIC^{\gets}:\ket{k,x,y,D} \mapsto\ket{k,x} \Fsf_k\cdot\Csf_{k,x}\cdot\Psf_{k,x}\cdot \Csf_{k,x}\cdot \Fsf_k\ket{y,D}\;.
\end{align*}
Here, $\Psf_{k,x}$ maps $\ket{k,x,y,D}\mapsto \ket{k,x,y\oplus D(k,x),D}$ if $D(k,x)\neq\bot$ and fixes the state otherwise; for every $(k,x)\in\Kcal\times\Xcal$ and $D$ such that $D(k,x)=\bot$, the operator $\Csf_{k,x}$ swaps the two vectors $\ket{+_{k,x,D}}$ and $\ket{D}$ and fixes the orthogonal complement, where
\[
\ket{+_{k,x,D}} := \frac{1}{\sqrt{2^n - \abs{D_k}}}\sum_{y\not\in \Im(D_k)}\ket{D[(k,x)\mapsto y]}\;;
\]
also the flipping operator $\Fsf_k$ maps $\ket{D}\mapsto\ket{D'}$ where $D'_{k'} := D_{k'}$ for every $k'\neq k$ and $D'_k = D_k^{-1}$.

As an application of our indistinguishability lifting theorem, we show that the compressed ideal cipher oracle as described above is indistinguishable from the ideal cipher $P^\Kcal$.
\begin{theorem}\label{thm:cIC}
    Let $\Acal$ be an oracle algorithm making at most $q$ quantum queries to either the compressed ideal cipher $\cIC$ or the ideal cipher $P^\Kcal$ bidirectionally, where the message space $\Xcal = \bool^n$ consists of all $n$-bit strings. Then, we have
    \[
    \SD\lrp{\Acal^\cIC,\Acal^{P^\Kcal}} \leq 8q^2\cdot \Delta_{\cpO}(n,q)\leq O\lrp{\frac{q^3}{2^{n/2}}}\;.
    \]
\end{theorem}
\begin{proof}
    Our lifting theorem (\cref{thm:ind_lift}) immediately implies another oracle algorithm $\Bcal$ making at most $q$ quantum queries such that
    \[
    \SD\lrp{\Acal^{\cpO^\Kcal},\Acal^{P^\Kcal}}
    \leq 8q^2\cdot \SD\lrp{\Bcal^{\cpO},\Bcal^{P}}\;.
    \]
    Hence, it suffices to argue that the keyed compressed permutation oracle $\cpO^\Kcal$, which runs one independent instance of $\cpO$ per key $k$, behaves identically to our previously defined compressed ideal cipher $\cIC$. This can be seen via considering the following isometry
    \[
    \Jcal:\ket{D}\mapsto \bigotimes_{k\in\Kcal}\ket{D_k}\;,
    \]
    which, by construction, has the easy-to-verify property that
    \begin{align*}
        &\Jcal\cdot \cIC^{\to} = \lrp{\cpO^\Kcal}^\to \cdot\Jcal\;,&
        \Jcal\cdot \cIC^{\gets} = \lrp{\cpO^\Kcal}^\gets \cdot\Jcal\;.
    \end{align*}
    This would imply that, for the hybrid game $\Hyb_i$ that replaces the first $i$ queries of $\Acal^{\cpO^\Kcal}$ to querying $\cIC$, we have that $\Hyb_{i}$ behaves identically to $\Hyb_{i+1}$ for every $i$. Hence $\Hyb_0$, which is simply $\Acal^{\cpO^\Kcal}$, behaves identically to $\Hyb_{q}$, which is $\Acal^{\cIC}$. We therefore have that
    \[
    \SD\lrp{\Acal^\cIC,\Acal^{P^\Kcal}} \leq 8q^2\cdot \Delta_{\cpO}(n,q)\leq O\lrp{\frac{q^3}{2^{n/2}}}\;,
    \]
    where the asymptotic bound is via plugging in the results that we have recalled in \cref{lem:Carolan}.
\end{proof}

\subsection{Preimage Resistance of Davies-Meyer Construction}

As an application of our compressed ideal cipher oracle, we show preimage resistance of the Davies-Meyer construction:
\[
\DM(L,R) := P^\Kcal(L,R) \oplus R\;,
\]
where we model $P^\Kcal:\Kcal\times\Xcal\to\Xcal$ as a (quantum) ideal cipher.

\begin{theorem}\label{thm:DM_preimage}
    Let $\Acal$ be an oracle algorithm making at most $q < \abs{\Xcal} - 1$ quantum queries to $P^\Kcal$ bidirectionally. Then for every $y\in\Xcal$, and for $(L,R)\gets\Acal^{P^\Kcal}(y)$, we have
    \[
    \Pr\sqbkt{\DM(L,R)=y}
    \leq \frac{32(q+1)^2 + 1}{\abs{\Xcal}-q-1} + 8(q+1)^2\Delta_{\cpO}(n,q+1)\;.
    \]
\end{theorem}
\begin{proof}
    First, we consider a game $\Gsf_0$. In this game, we run $\Acal$, but substitute each forward query $(k,x)$ to $P^\Kcal$ with a query to $P^\Kcal(k,x)\oplus S$, where $S\gets\Xcal$ is sampled at the beginning of the game. Likewise, we replace the backward queries in the obvious way.
    After the output $(L,R)$ is produced by $\Acal$, the game aborts if $R = Y_S := y\oplus S$ and otherwise outputs the result of the check whether $P^\Kcal(L,R) \oplus R = Y_S$. The substitution to each query has no effect on the output distribution of $\Acal$, and since $S$ is kept secret, the abort only happens with probability at most $1/\abs{\Xcal}$. Therefore,
    \begin{equation}
        \Pr_{(L,R)\gets\Acal^{P^\Kcal}}\sqbkt{\DM(L,R) = y}
        \leq \Pr\sqbkt{\Gsf_0} + 1/\abs{\Xcal}\;.
        \label{eq:DM_preimage_G0}
    \end{equation}

    Next, we consider a game $\Gsf_1$, in which we further substitute all queries to $P^\Kcal$ in $\Gsf_0$ with queries to the compressed ideal cipher oracle $\cIC$. Due to our prior proven indistinguishability of compressed ideal cipher, we have that
    \begin{equation}\label{eq:preimage_proof_ind}
        \Pr\sqbkt{\Gsf_0}
        \leq \Pr\sqbkt{\Gsf_1} + 8(q+1)^2\Delta_{\cpO}(n,q+1)\;,
    \end{equation}
    where the $+1$ in query complexity comes from the final verification query.
    In this game, let $\ket{\phi_0}_{\Xsf\Ysf\Zsf\Dsf}$ be the initial state of $\Acal$'s query input and output registers $\Xsf$ and $\Ysf$, working register $\Zsf$, and the database register $\Dsf$ of $\cIC$, which is just a constant-$\bot$ database. Moreover, let $\ket{\phi_i}$ be the state of $\Xsf\Ysf\Zsf\Dsf$ right before the $i$th query. Then for every fixed choice $y_s$ of $Y_S$,
    the projector
    $$
    \Pi^{y_s} := \sum_{k,x:x\neq y_s}\ket{(k,x),y_s\oplus x}\bra{(k,x),y_s\oplus x}_{\Xsf\Ysf}
    $$
    would capture the final successful event of $\Acal$ outputting a preimage of $y_s$, and the game does not abort. Therefore, we have that
    \begin{equation}\label{eq:G0_preimage}
        \Pr\sqbkt{\Gsf_1\mid Y_S = y_s}
        \leq \norm{\Pi^{y_s} \cdot\cIC^\to \ket{\phi_{q+1}}_{\Xsf\Ysf\Zsf\Dsf}}^2
    \end{equation}
    where we note that $\ket{\phi_{q+1}}$ can be written in the form of $\sum_{k,x} \alpha_{k,x}\ket{k,x,0}\ket{\psi_{k,x}}_{\Zsf\Dsf}$, so that $(\ket{0}\bra{0})_\Ysf \ket{\phi_{q+1}} = \ket{\phi_{q+1}}$.
    
    Similar to \cite{EC:CFHL21}, define the following sets of databases
    \begin{align*}
        &\SZ_i := \set{D\;\big|\;\abs{D}\leq i}\;,
        &\PRMG_{y_s} := \set{D\;\big|\;\exists ((k,x),w)\in D: w\oplus x = {y_s}}\;,
    \end{align*}
    where we then abuse the notation and treat these sets as projectors onto the databases in these sets respectively.

    Throughout the game, since $\Acal$ makes at most $q$ quantum queries, it follows that $\ket{\phi_i}$ is supported by all states of database size at most $q$, i.e.\ $\SZ_q \ket{\phi_i} = \ket{\phi_i}$. 
    For every $(k,x)\in\Kcal\times\Xcal$, let
    $D$ be a database such that $D(k,x) = \bot$, and let
    $\cIC^\to_{k,x,D}$ be the restriction of $\cIC^\to$ to the invariant subspace where $\Xsf$ is fixed in the computational basis to $\ket{k,x}$, and $\Dsf$ is supported by database states $\ket{D'}$ such that $D'(k',x') = D(k',x')$ for every $(k',x')\neq (k,x)$. Given that all entries in $D'$ other than $D'(k,x)$ are fixed, the state space of $\Dsf$ is spanned by all possible $\ket{D'(k,x)}$, in which we have:
    \ifarxivstyle
        \begin{align}
            \norm{\PRMG_{y_s}\cdot\Csf_{k,x}\cdot\SZ_q \cdot(\Ibb-\PRMG_{y_s})}^2 &\leq \norm{\sum_{w\in\Xcal\cup\set{\bot}\setminus\set{y_s\oplus x}\atop w\not\in\Im(D_k)}\ket{y_s\oplus x}\bra{y_s\oplus x}\Csf_{k,x}\cdot\SZ_q \ket{w}\bra{w}}^2\notag\\
            &= \sum_{w\in\Xcal\cup\set{\bot}\setminus\set{y_s\oplus x}\atop w\not\in\Im(D_k)}\abs{\bra{y_s\oplus x}\Csf_{k,x}\cdot\SZ_q \ket{w}}^2\notag\\
            &\leq \idx_{D\in\SZ_q}\cdot \lrp{\frac{1}{\abs{\Xcal}-\abs{D_k}} + \frac{\abs{\Xcal}-\abs{D_k}}{(\abs{\Xcal}-\abs{D_k})^2}}
            \leq \frac{2}{\abs{\Xcal}-q}\;,
            \label{eq:compression_transition_bound}
        \end{align}
    \else
        \begin{align}
            &\norm{\PRMG_{y_s}\cdot\Csf_{k,x}\cdot\SZ_q \cdot(\Ibb-\PRMG_{y_s})}^2\notag\\
            &\leq \norm{\sum_{w\in\Xcal\cup\set{\bot}\setminus\set{y_s\oplus x}\atop w\not\in\Im(D_k)}\ket{y_s\oplus x}\bra{y_s\oplus x}\Csf_{k,x}\cdot\SZ_q \ket{w}\bra{w}}^2\notag\\
            &= \sum_{w\in\Xcal\cup\set{\bot}\setminus\set{y_s\oplus x}\atop w\not\in\Im(D_k)}\abs{\bra{y_s\oplus x}\Csf_{k,x}\cdot\SZ_q \ket{w}}^2\notag\\
            &\leq \idx_{D\in\SZ_q}\cdot \lrp{\frac{1}{\abs{\Xcal}-\abs{D_k}} + \frac{\abs{\Xcal}-\abs{D_k}}{(\abs{\Xcal}-\abs{D_k})^2}}
            \leq \frac{2}{\abs{\Xcal}-q}\;,
            \label{eq:compression_transition_bound}
        \end{align}
    \fi
    where the last inequality is via $\abs{D_k}\leq \abs{D}\leq q$ for every $D\in\SZ_q$, and
    the second to last inequality is via observing the fact that
    \[
        \abs{\bra{y_s\oplus x}\Csf_{k,x}\cdot\SZ_q\ket{\bot}}
        \leq \frac{1}{\sqrt{\abs{\Xcal}-\abs{D_k}}}\;,\qquad
        \abs{\bra{y_s\oplus x}\Csf_{k,x}\cdot\SZ_q\ket{w}}
        \leq \frac{1}{\abs{\Xcal}-\abs{D_k}}\;.
    \]
    Now, observing that $\Psf_{k,x}$ commutes with both $\SZ_q$ and $\PRMG_{y_s}$, and that $(\Ibb-\SZ_{q+1})\cdot\Csf_{k,x}\cdot\SZ_q = 0$, triangle inequality would give us that
    \begin{align}
        &\norm{\PRMG_{y_s}\cdot \cIC^{\to}_{k,x} \cdot\SZ_q \cdot(\Ibb-\PRMG_{y_s})}\notag\\
        &\leq \norm{\PRMG_{y_s}\cdot \Csf_{k,x} \cdot\SZ_q \cdot(\Ibb-\PRMG_{y_s})} + \norm{\PRMG_{y_s}\cdot \Csf_{k,x} \cdot\SZ_{q+1} \cdot(\Ibb-\PRMG_{y_s})}\notag\\
        &\leq \sqrt{\frac{2}{\abs{\Xcal}-q}} + \sqrt{\frac{2}{\abs{\Xcal}-q-1}} \leq \sqrt{\frac{8}{\abs{\Xcal}-q-1}}\;.
        \label{eq:transition_forward}
    \end{align}
    Likewise, by the symmetry of $\cIC$, we can similarly define $\cIC^{\gets}_{k,x,D}$ in the obvious way, and have
    \begin{equation}\label{eq:transition_back}
        \norm{\PRMG_{y_s}\cdot \cIC^{\gets}_{k,x,D} \cdot\SZ_q \cdot(\Ibb-\PRMG_{y_s})}\leq \sqrt{\frac{8}{\abs{\Xcal}-q-1}}\;.
    \end{equation}
    
    Combining \eqref{eq:transition_forward} and \eqref{eq:transition_back} and the fact that $\SZ_q\ket{\phi_i} = \ket{\phi_i}$, we have that
    \ifarxivstyle
        \begin{align*}
            \norm{\PRMG_{y_s}\ket{\phi_{i+1}}}&\leq \sum_{d\in\set{\to,\gets}}\norm{\PRMG_{y_s}\cdot \lrp{\cIC^{d}\otimes\ket{d}\bra{d}}\cdot(\Ibb-\PRMG_{y_s})\ket{\phi_{i}}}+
            \norm{\PRMG_{y_s} \ket{\phi_i}}\\
            &\leq \sqrt{\frac{32}{\abs{\Xcal} - q -1}} + \norm{\PRMG_{y_s} \ket{\phi_i}} 
        \end{align*}
    \else
        \begin{align*}
            &\norm{\PRMG_{y_s}\ket{\phi_{i+1}}}\\
            &\leq \sum_{d\in\set{\to,\gets}}\norm{\PRMG_{y_s}\cdot \lrp{\cIC^{d}\otimes\ket{d}\bra{d}}\cdot(\Ibb-\PRMG_{y_s})\ket{\phi_{i}}}+
            \norm{\PRMG_{y_s} \ket{\phi_i}}\\
            &\leq \sqrt{\frac{32}{\abs{\Xcal} - q -1}} + \norm{\PRMG_{y_s} \ket{\phi_i}} 
        \end{align*}
    \fi
    where the projector $\ket{d}\bra{d}$ acts on a designated register where $\Acal$ indicates the direction of a $\cIC$ query. Applying the inequality recursively would then give us
    \begin{equation}\label{eq:PRMG_q_plus_1}
        \norm{\PRMG_{y_s}\ket{\phi_{q+1}}}\leq q\sqrt{\frac{32}{\abs{\Xcal}-q-1}}\;.
    \end{equation}
    Therefore, in the invariant subspace of every fixed $k,x$ again, we obtain
    \ifarxivstyle
        \begin{align}
            &\norm{\Pi^{y_s} \cdot\cIC^\to \cdot\SZ_q \cdot(\Ibb-\PRMG_{y_s})\cdot(\ket{0}\bra{0})_\Ysf}\notag\\
            &\leq \norm{\Pi^{y_s} \cdot (\Csf_{k,x}\cdot\Psf_{k,x}\cdot\Csf_{k,x})\cdot\SZ_q  \cdot(\Ibb-\PRMG_{y_s})\cdot(\ket{0}\bra{0})_\Ysf}\notag\notag\\
            &\leq \norm{\Pi^{y_s} \cdot \Psf_{k,x}\cdot(\Ibb-\PRMG_{y_s})\cdot\SZ_q\cdot (\ket{0}\bra{0})_\Ysf} + \sqrt{\frac{8}{\abs{\Xcal}-q-1}}=\sqrt{\frac{8}{\abs{\Xcal}-q-1}}\;,\label{eq:preimage_final_invariant}
        \end{align}
    \else
        \begin{align}
            &\norm{\Pi^{y_s} \cdot\cIC^\to \cdot\SZ_q \cdot(\Ibb-\PRMG_{y_s})\cdot(\ket{0}\bra{0})_\Ysf}\notag\\
            &\leq \norm{\Pi^{y_s} \cdot (\Csf_{k,x}\cdot\Psf_{k,x}\cdot\Csf_{k,x})\cdot\SZ_q  \cdot(\Ibb-\PRMG_{y_s})\cdot(\ket{0}\bra{0})_\Ysf}\notag\notag\\
            &\leq \norm{\Pi^{y_s} \cdot \Psf_{k,x}\cdot(\Ibb-\PRMG_{y_s})\cdot\SZ_q\cdot (\ket{0}\bra{0})_\Ysf} + \sqrt{\frac{8}{\abs{\Xcal}-q-1}}\notag\\
            &= \sqrt{\frac{8}{\abs{\Xcal}-q-1}}\;,\label{eq:preimage_final_invariant}
        \end{align}
    \fi
    where the second inequality is via \eqref{eq:compression_transition_bound} and the fact that $\Pi^{y_s}$ and $(\ket{0}\bra{0})_\Ysf$ commute with $\Csf_{k,x}$, and the last equality is via the fact that $\Pi^{y_s}\cdot\Psf_{k,x}\cdot(\Ibb-\PRMG_{y_s})\cdot (\ket{0}\bra{0})_\Ysf = 0$. Therefore we get
    \ifarxivstyle
        \begin{align*}
            &\norm{\Pi^{y_s} \cdot\cIC^\to \ket{\phi_{q+1}}}
            \leq  \norm{\Pi^{y_s}\cdot\cIC^\to\cdot\SZ_q\ket{\phi_{q+1}}}\\
            &\leq  \norm{\Pi^{y_s}\cdot\cIC^\to\cdot\SZ_q \cdot(\Ibb-\PRMG_{y_s})\ket{\phi_{q+1}}}+ q\sqrt{\frac{32}{\abs{\Xcal}-q-1}}\leq (2q+1)\sqrt{\frac{8}{\abs{\Xcal}-q-1}}
        \end{align*}
    \else
        \begin{align*}
            &\norm{\Pi^{y_s} \cdot\cIC^\to \ket{\phi_{q+1}}}
            \leq  \norm{\Pi^{y_s}\cdot\cIC^\to\cdot\SZ_q\ket{\phi_{q+1}}}\\
            &\leq  \norm{\Pi^{y_s}\cdot\cIC^\to\cdot\SZ_q \cdot(\Ibb-\PRMG_{y_s})\ket{\phi_{q+1}}}+ q\sqrt{\frac{32}{\abs{\Xcal}-q-1}}\\
            &\leq (2q+1)\sqrt{\frac{8}{\abs{\Xcal}-q-1}}
        \end{align*}
    \fi
    where the first inequality is via $\SZ_q\ket{\phi_{q+1}} = \ket{\phi_{q+1}}$, the second inequality is via \eqref{eq:PRMG_q_plus_1}, and the last inequality is via $(\ket{0}\bra{0})_\Ysf \ket{\phi_{q+1}} = \ket{\phi_{q+1}}$ and
    \eqref{eq:preimage_final_invariant}. Combining the above with \eqref{eq:DM_preimage_G0}, \eqref{eq:preimage_proof_ind} and \eqref{eq:G0_preimage}, the proof is concluded.
\end{proof}

\section{Domain Extension of Quantum PRP}\label{sec:domain_ext}

As an application of our generic lifting theorem, we show that, given a quantum-secure PRP $\PRP_K:\bool^n\to\bool^n$ that operates on $n$-bit messages, where the subscript $K$ denotes a PRP key of length $\lambda$, there is a quantum-secure PRP $\PRP'_K$ that operates on $2n$-bit messages. 

\begin{theorem}
    Let $\PRP_K:\bool^n\to\bool^n$ be a quantum-secure strong PRP of key length $\lambda$. Then there is a quantum-secure strong PRP $\PRP'_{K_3}:\bool^{2n}\to\bool^{2n}$ of key length $2\lambda\ceil{\lambda/n}$, which makes at most $2\ceil{\lambda/n}+2$ queries to $\PRP$, such that the following holds.
    For every attacker $\Acal$ making at most $q$ bidirectional quantum queries, there is a similarly efficient attacker $\Bcal$ making at most $6q$ quantum bidirectional queries such that
    \[
    \Adv^{\RP^{\leftrightarrow}}_{\PRP'}(\Acal)
    \leq \lrp{72q^2+2\ceil{\lambda/n}}\Adv^{\RP^{\leftrightarrow}}_{\PRP}(\Bcal)
    + \frac{1 + 11664\ceil{\lambda/n} q^3 + 32770 q^4}{2^n}\;.
    \]
\end{theorem}

We achieve this via a three-step construction, where the preparation step follows from existing results, and each of the remaining two steps is via a generic transformation that may be of independent interest.

\subsubsection{Preparation: from PRP to PRF.} First, we derive a quantum-secure PRF $\PRF_{K_1}:\bool\times\bool^n\to\bool^\lambda$ of key length $\lambda_1 := 2\lambda\ceil{\lambda/n}$ bits.
This step is trivial:
for $\ell := \ceil{\lambda/n}$ and $2\ell$ independent $\PRP$ keys $K_1 := (K_{1,b,i})_{b\in\bool,i\in[\ell]}$, we simply concatenate $\ell$ instances of $\PRP$,
\begin{equation}
    \PRF_{K_1}(b,X) := \PRP_{K_{1,b,1}}(X)\Vert\dots\Vert \PRP_{K_{1,b,\ell}}(X)\;,
\end{equation}
and then truncate the output to $\lambda$ bits. Zhandry's PRP-PRF switching lemma \cite{FOCS:Zhandry12} implies this is indeed a secure PRF. For every attacker making $q$ quantum queries, the precise bound of this step is
\begin{equation}\label{eq:PRF_from_PRP}
    \Adv^{\RO}_{\PRF}(\Acal)
    \leq 2\ceil{\lambda/n}\Adv^{\RP^{\leftrightarrow}}_{\PRP}(\Bcal) + \frac{54\ceil{\lambda/n}q^3}{2^n}\;,
\end{equation}
for a similarly efficient attacker $\Bcal$ making at most $q$ quantum queries, where the constant $27$ in each of the $2\ceil{\lambda/n}$ switching steps comes from the constant $(8\pi^2/3)$ spelled out concretely in \cite[Theorem~1]{CiC:Xagawa24}, and the factor $2\ceil{\lambda/n}$ comes from a hybrid argument over the $2\ceil{\lambda/n}$ independent keys.

\subsubsection{Step $1$: from PRP and PRF to tweakable cipher.} Second, we derive a quantum-secure tweakable cipher $\TWK_{K_2}:(\bool\times\bool^n)\times\bool^n\to\bool^n$ with tweak space $\bool\times\bool^n$ and key length $\lambda_2 := \lambda_1$ bits, from $\PRP$ and $\PRF$ as follows:
\begin{equation}\label{eq:TWK}
    \TWK_{K_2}(t,x) := \PRP_{\PRF_{K_1}(t)}(x)\;,
\end{equation}
where $K_2$ is simply the PRF key $K_1$.
It simply uses $\PRF$ to generate a key for each tweak $t$, and the goal here is to show that this construction is indistinguishable from a quantum ideal cipher oracle, which we will argue in \cref{sec:generic_tweakable} using our indistinguishability lifting theorem \cref{thm:ind_lift}.

\subsubsection{Step $2$: from tweakable cipher to PRP of double message length.} Finally, we derive a quantum-secure PRP $\PRP'_{K_3}:\bool^{2n}\to\bool^{2n}$ of key length $\lambda_3 := \lambda_2$, from $\TWK:(\bool\times\bool^n)\times\bool^n\to\bool^n$ as follows:
\begin{equation}\label{eq:extPRP}
    \PRP'_{K_3}(L,R) := (\TWK_{K_3}((1,U),L),U)\;,\quad
    U := \TWK_{K_3}((0,L),R)\;,
\end{equation}
where $K_3$ is a single $\TWK$ key, and the two rounds are separated by the first bit of the tweak. In \cref{sec:tweak2ext}, we show via the polynomial method that the construction above is quantum indistinguishable from a random permutation.

Combining \eqref{eq:PRF_from_PRP} with \cref{lem:generic_tweakable,lem:domain_extension} below proves the theorem. Indeed, the attacker against $\TWK$ in \cref{lem:domain_extension} makes at most $3q$ queries, so the attacker against $\PRF$ in \cref{lem:generic_tweakable} makes at most $6q$ queries and the attacker against $\PRP$ in \cref{lem:generic_tweakable} makes at most $3q$ queries. This gives
\begin{align*}
    \Adv^{\RP^{\leftrightarrow}}_{\PRP'}(\Acal)
    &\leq 2\ceil{\lambda/n}\Adv^{\RP^{\leftrightarrow}}_{\PRP}(\Bcal) + \frac{54\ceil{\lambda/n}(6q)^3}{2^n}\\
    &\qquad+ 8(3q)^2\Adv^{\RP^{\leftrightarrow}}_{\PRP}(\Bcal) 
    + 2^{-n}
    + \frac{32770q^4}{2^n}\;,
\end{align*}
where $\Bcal$ is the one of the two $\PRP$ attackers with the larger advantage, and makes at most $6q$ queries. The key length is $\lambda_3=\lambda_1=2\lambda\ceil{\lambda/n}$, and each evaluation of $\PRP'$ makes two evaluations of $\TWK$, each consisting of $\ceil{\lambda/n}$ calls to $\PRP$ for $\PRF$ and one further call to $\PRP$.

\begin{remark}
    For simplicity we only consider doubling the message-length. It is possible to further extend our technique toward extending input length to $k$ times.
    We simply need a $\PRF$ of input size of $O(kn)$ bits, which would allow us to construct the tweakable cipher $\TWK$ of tweak length $O(kn)$ bits in Step~$1$, thereby allowing us to perform a recursion of depth $O(\log k)$ in Step~$2$.
\end{remark}

\subsection{A Generic Tweakable Cipher Construction}\label{sec:generic_tweakable}
In the statement below, let $\PRP_K:\bool^{n}\to\bool^n$ be a quantum-secure strong PRP of key length $\lambda$ bits, and $\PRF_{K_1}:\bool^m\to\bool^{\lambda}$ be a quantum secure $\PRF$ of key length $\lambda_1$ bits. (In our application, $m=n+1$, where we identify $\bool\times\bool^n$ with $\bool^{n+1}$.) We show that the tweakable cipher construction in \eqref{eq:TWK}:
\[
\TWK_{K_2}(t,x) := \PRP_{\PRF_{K_1}(t)}(x)\;,
\]
where $K_2 := K_1\gets\bool^{\lambda_1}$, is also quantum-secure.

\begin{lemma}\label{lem:generic_tweakable}
    For every oracle algorithm $\Acal$ making at most $q$ quantum bidirectional queries, there exist similarly efficient oracle algorithms $\Bcal$ and $\Ccal$ making at most $2q$ and $q$ quantum queries respectively, such that
    \[
    \Adv^{\IC^{\leftrightarrow}}_{\TWK}(\Acal)
    \leq \Adv^{\RO}_{\PRF}(\Bcal) + 8q^2\cdot\Adv^{\RP^{\leftrightarrow}}_{\PRP}(\Ccal) + 2^{-n}\;.
    \]
    Here, queries made by $\Acal$ and $\Ccal$ are allowed to be bidirectional.
\end{lemma}
\begin{proof}
    Let $H:\bool^m\to\bool^\lambda$ be a random oracle, $P:\bool^n\to\bool^n$ be a random permutation, and $P^{\bool^m}:\bool^m\times\bool^n\to\bool^n$ be an ideal cipher. The goal of our proof is to show that, for $K_2:=K_1\gets\bool^{\lambda_1}$, we have
    \[
    \Adv^{\IC^{\leftrightarrow}}_{\TWK}(\Acal)
    = \SD\lrp{\Acal^{\TWK_{K_2}}, \Acal^{P^{\bool^m}}}
    \]
    To that end, consider an intermediate hybrid oracle, defined as follows
    \[
    \Hyb(t,x) := \PRP_{H(t)}(x)\;,
    \]
    with inverse queries defined in the obvious way.
    Now, immediately from the quantum security of $\PRF$, we can switch out the $\PRF$ underlying $\TWK$ and hence obtain
    \[
    \SD\lrp{\Acal^{\TWK_{K_2}}, \Acal^{\Hyb}}
    \leq \Adv^{\RO}_{\PRF}(\Bcal)\;,
    \]
    for a similarly efficient oracle algorithm $\Bcal$. Since in each query that $\Acal$ makes, $\Bcal$ would need to make $2$ queries to the underlying $\PRF$/$H$ oracle, it sums up to $2q$ queries in total.
    
    Next, we note that the oracle $\Hyb$ can be regarded as the oracle $\PRP^{\bool^m}$ that samples an exponential family of $\PRP$ keys $K_{t}\gets \bool^{\lambda}$, one for each tweak $t$ independently and then responds with $\PRP_{K_{t}}(x)$ on input $(t,x)$. Hence, by applying our indistinguishability lifting theorem (\cref{thm:ind_lift}), we obtain
    \begin{align*}
        &\SD\lrp{\Acal^{\Hyb}, \Acal^{P^{\bool^m}}}
        =\SD\lrp{\Acal^{\PRP^{\bool^m}}, \Acal^{P^{\bool^m}}}\\
        &\leq 8q^2\cdot \SD\lrp{\Dcal^{\PRP, \PRP^{\bool^m},P^{\bool^m}}, \Dcal^{P, \PRP^{\bool^m},P^{\bool^m}}}\\
        &\leq 8q^2\cdot \lrp{\SD\lrp{\Ccal^{\PRP}, \Ccal^{P}}+2\cdot 2^{-n}/(16q^2)}
        = 8q^2\cdot \Adv^{\RP^{\leftrightarrow}}_{\PRP}(\Ccal) + 2^{-n}\;,
    \end{align*}
    for a similarly efficient distinguishing algorithm $\Dcal$, and the final reduction algorithm $\Ccal$ that internally simulates the oracle queries to $\PRP^{\bool^m}$ and $P^{\bool^m}$ up to a small additive error. Concretely, similar to what was pointed out in \cite[Appendix~A]{EPRINT:HHMRS26}, we can simulate a random permutation up to a $\delta$ error within a runtime overhead of order $O(q^2\cdot n^2\lrp{n+\log(1/\delta)}^2)$, counting basic arithmetics as unit runtime.
    Picking $\delta := 2^{-(n+m)}/(16q^2)$ would make the simulation of the entire family $P^{\bool^{m}}$ of error bounded by $2^{-n}/(16q^2)$, within a running time overhead of order $O(q^2n^2(n+m+\log q)^2)$.

    Collecting the bounds so far, the proof is concluded.
\end{proof}

\subsection{From Tweakable Cipher to Double-Input PRP}\label{sec:tweak2ext}

In the statement below, let $\TWK_{K_3}:(\bool\times\bool^{n})\times\bool^{n}\to\bool^n$ be a quantum-secure tweakable cipher with tweak space $\bool\times\bool^n$ and key length $\lambda_3$ bits. We show that the strong PRP construction described in \cref{eq:extPRP}, which we recall below:
\[
    \PRP'_{K_3}(L, R) := (\TWK_{K_3}((1,U),L),U)\;,\quad U:=\TWK_{K_3}((0,L),R)\;,
\]
is also quantum-secure. Here, both rounds use the same key $K_3$, and are separated by the first bit of the tweak.
\begin{lemma}
    For every oracle algorithm $\Acal$ making at most $q$ quantum bidirectional queries, there exists a similarly efficient oracle algorithm $\Bcal$ making at most $3q$ quantum bidirectional queries, such that
    \[
        \Adv^{\RP^{\leftrightarrow}}_{\PRP'}(\Acal)
        \leq \Adv^{\IC^{\leftrightarrow}}_{\TWK}(\Bcal) +  \frac{32770q^4}{2^n}\;.
    \]
    \label{lem:domain_extension}
\end{lemma}
\begin{proof}
    Let $\Bcal^{E}$ run $\Acal$ and answer each of its queries by evaluating $\PRP'$ with $\TWK_{K_3}$ replaced by its oracle $E$. A forward query $(L,R)$ is answered by computing $U:=E((0,L),R)$ into an ancilla, XOR-ing $E((1,U),L)$ and $U$ into the answer register, and uncomputing $U$; an inverse query $(V,U)$ is answered symmetrically via $L:=E^{-1}((1,U),V)$ and $R=E^{-1}((0,L),U)$. Since the direction bit only selects which tweak, input, and direction are loaded into the query registers, $\Bcal$ makes $3$ bidirectional queries to $E$ per query of $\Acal$, and has a similar running time as $\Acal$ does. Clearly $\Bcal^{\TWK_{K_3}}$ behaves identically to $\Acal^{\PRP'_{K_3}}$. On the other hand, let $\IC\gets\Sym(\bool^n)^{\bool\times\bool^n}$ be an ideal cipher with tweak space $\bool\times\bool^n$, and let $G(L,R):=\IC((0,L),R)$ and $H(U,L):=\IC((1,U),L)$. Then $G$ and $H$ are two independent ideal ciphers over the message space $\bool^n$ and tweak space $\bool^n$, and $\Bcal^{\IC}$ behaves identically to $\Acal^{F_{G,H}}$, where
    \[
        F_{G, H}(L, R) := (H(G(L, R), L), G(L, R))\;.
    \]
    Therefore, for a random permutation $\RP\gets\Sym(\bool^{2n})$, the triangle inequality gives
    \begin{align}
        \SD\lrp{\Acal^{\PRP'_{K_3}}, \Acal^{\RP}}
        &\leq \Adv^{\IC^{\leftrightarrow}}_{\TWK}(\Bcal) + \SD\lrp{\Acal^{F_{G,H}},\Acal^{\RP}}\;.\label{eq:standard_hybrid_tweakable}
    \end{align}
    For the remainder of the proof we focus on bounding the last term of \eqref{eq:standard_hybrid_tweakable}. We will show that this term is bounded by $32770q^4/2^n$, which will complete the proof.

    \paragraph{Output probability for a fixed permutation.}
    We first observe that there is a multilinear polynomial $a$ of degree at most $2q$ in the variables $(\idx_{\pi(x)=y})_{x,y\in\bool^{2n}}$ such that, for every fixed permutation $\pi\in\Sym(\bool^{2n})$,
    \[
        \Pr\sqbkt{\Acal^{\pi}} = a\lrp{(\idx_{\pi(x)=y})_{x,y}}\;.
    \]
    Indeed, each query implements the unitary
    \[
        \ket{x,y}\mapsto\ket{x,y\oplus \pi^{\pm 1}(x)}
        = \sum_{z}\idx_{\pi^{\pm 1}(x)=y\oplus z}\ket{x,z}\;,
    \]
    where for inverse queries $\idx_{\pi^{-1}(x)=y\oplus z}=\idx_{\pi(y\oplus z)=x}$, so the entries of this unitary are linear in these variables. Hence, after $q$ queries, every amplitude is a polynomial of degree at most $q$, and the output probability, being a sum of squared magnitudes of amplitudes, is a polynomial of degree at most $2q$. Since the variables only take values in $\bool$, we may take this polynomial to be multilinear.

    \paragraph{Output probability for a rerandomized permutation.}
    Now, consider the following rerandomization procedure $\pi'\gets\Rsf(\pi)$: it independently samples $P,Q\gets\Sym(\bool^n)^{\bool^n}$ and produces
    \[
        \pi'(L,R) := (Q(U,V),U)\;,\qquad (V,U):=\pi(L,P(L,R))\;.
    \]
    In other words, $\pi'=\tau\circ\pi\circ\sigma$ for the random permutations $\sigma(L,R):=(L,P(L,R))$ and $\tau(V,U):=(Q(U,V),U)$ of $\bool^{2n}$. We have
    \begin{equation}
        \Pr\sqbkt{\Acal^{\RP}} = \Pr\sqbkt{\Acal^{\Rsf(\RP)}}\;,\qquad \Pr\sqbkt{\Acal^{F_{G,H}}} = \Pr\sqbkt{\Acal^{\Rsf(F_{G,H})}}\;.\label{eq:invariance_rerandomize}
    \end{equation}
    Indeed, $\tau\circ\RP\circ\sigma$ is again a uniformly random permutation, since $\sigma$ and $\tau$ are independent of $\RP$. Moreover, $\Rsf(F_{G,H})=F_{G',H'}$, where $G'(L,\cdot):=G(L,\cdot)\circ P(L,\cdot)$ and $H'(U,\cdot):=Q(U,\cdot)\circ H(U,\cdot)$, and $G'$ and $H'$ are again two independent ideal ciphers.

    Next, for every fixed permutation $\pi\in\Sym(\bool^{2n})$ and for $\pi'\gets\Rsf(\pi)$, we have
    \begin{align}
        \Pr\sqbkt{\Acal^{\Rsf(\pi)}}
        &= \expn\sqbkt{a\lrp{(\idx_{\pi'(x)=y})_{x,y}}}\notag\\
        &= \expn\sqbkt{\sum_{S} \widehat{a}(S)\prod_{(x,y)\in S} \idx_{\pi'(x)=y}}\notag\\
        &= \sum_{S} \widehat{a}(S) \cdot\Pr\sqbkt{\forall ((L,R),(V,U))\in S\atop\pi'(L,R)=(V,U)}\;,\label{eq:multilinear_poly_rerandomize}
    \end{align}
    where the first equality applies the previous paragraph to $\pi'$, the second expands $a$ into monomials, and the third uses linearity of expectation. Here, the summation runs over all
    \[
        S = \set{(x_1,y_1),\dots,(x_k,y_k)}\subseteq\bool^{2n}\times\bool^{2n}
    \]
    of size at most $2q$ such that the $x_i$'s are all distinct and the $y_i$'s are all distinct,\footnote{Note that the probability in \eqref{eq:multilinear_poly_rerandomize} would have vanished if $S$ cannot be completed to a permutation.} and $\widehat{a}(S)\in\Rbb$ is the coefficient corresponding to the monomial $\prod_{(x,y)\in S}\idx_{\pi'(x)=y}$.

    \paragraph{A matrix $C^\pi$ for counting.}
    Towards calculating \eqref{eq:multilinear_poly_rerandomize}, for every fixed permutation $\pi\in\Sym(\bool^{2n})$, let $C^\pi$ be a matrix whose rows and columns are indexed by $L\in\bool^n$ and $U\in\bool^n$, respectively, with each entry $C^\pi_{L,U}$ defined as follows:
    \[
        C^\pi_{L,U} := \abs{\set{R\in\bool^n\;\big|\;\pi(L,R)\in\bool^n\times\set{U}}}\;.
    \]
    Now fix a set $S=\set{(x_1,y_1),\dots,(x_k,y_k)}$ as in \eqref{eq:multilinear_poly_rerandomize}, and write $(L_i,R_i) := x_i$ and $(V_i,U_i) := y_i$ for each $i$ (where $L_i,R_i,V_i,U_i\in\bool^n$), and let
    \[
        e_{L,U}:=\abs{\set{i\;\big|\;(L,U) = (L_i,U_i)}}\;,\quad
        \rowsum{e}{L}:=\sum_U e_{L,U}\;,\quad \colsum{e}{U}:=\sum_L e_{L,U}\;.
    \]
    We will show that
    \begin{equation}\label{eq:rerandomized_partial_permutation}
        \Pr_{\pi'\gets\Rsf(\pi)}\sqbkt{\forall ((L,R),(V,U))\in S\atop\pi'(L,R)=(V,U)}
        = \frac{\prod_{L,U\in\bool^n}(C^\pi_{L,U})_{e_{L,U}}}
               {\prod_{L\in\bool^n}(2^n)_{\rowsum{e}{L}}\prod_{U\in\bool^n}(2^n)_{\colsum{e}{U}}}\;.
    \end{equation}
    For each $i$, let $\tilde R_i:=P(L_i,R_i)$, and let $\tilde V_i$ denote the first half of $\pi(L_i,\tilde R_i)$. By definition of $\pi'$, we have $\pi'(x_i)=y_i$ if and only if the second half of $\pi(L_i,\tilde R_i)$ equals $U_i$ and $Q(U_i,\tilde V_i)=V_i$. First fix $L\in\bool^n$. The $R_i$ with $L_i=L$ are distinct and $P(L,\cdot)$ is a uniform permutation, so $(\tilde R_i)_{i:L_i=L}$ is a uniformly random sequence of $\rowsum{e}{L}$ distinct elements of $\bool^n$. Among the $(2^n)_{\rowsum{e}{L}}$ such sequences, exactly $\prod_U(C^\pi_{L,U})_{e_{L,U}}$ satisfy $\pi(L,\tilde R_i)\in\bool^n\times\set{U_i}$ for all $i$ with $L_i=L$, since for each $U$ the $e_{L,U}$ values $\tilde R_i$ with $U_i=U$ must be distinct elements of the set counted by $C^\pi_{L,U}$. Next fix $U\in\bool^n$, and condition on $P$. The pairs $(L_i,\tilde R_i)$ are distinct, so the values $\tilde V_i$ with $U_i=U$ are distinct since $\pi$ is injective, and so are the $V_i$ with $U_i=U$ since the $y_i$ are distinct. Since $Q(U,\cdot)$ is a uniform permutation independent of $P$, the event $Q(U,\tilde V_i)=V_i$ for all $i$ with $U_i=U$ has probability $1/(2^n)_{\colsum{e}{U}}$. Multiplying over all $L$ and $U$, using the independence of the permutations $P(L,\cdot)$ and $Q(U,\cdot)$, gives \eqref{eq:rerandomized_partial_permutation}.

    Therefore, combined with \eqref{eq:multilinear_poly_rerandomize}, we have
    \begin{align}
        \Pr\sqbkt{\Acal^{\Rsf(\pi)}}
        &= \sum_{S} \widehat{a}(S) \cdot \frac{\prod_{L,U\in\bool^n}(C^\pi_{L,U})_{e_{L,U}}}{\prod_{L\in\bool^n}(2^n)_{\rowsum{e}{L}}\prod_{U\in \bool^n}(2^n)_{\colsum{e}{U}}}
        =: g(C^{\pi})\;,\label{eq:def_g_block_matrix}
    \end{align}
    which is a polynomial of degree at most $2q$ in the entries of $C^{\pi}$, since $\sum_{L,U}e_{L,U}=\abs{S}\leq 2q$.

    \paragraph{Some useful properties of $C^\pi$.}
    Before moving ahead, we establish some basic facts about the matrix $C^\pi$, whose proofs are deferred to \cref{sec:remaining_claim_proofs}.

    \begin{restatable}{claim}{claimMatrixWhole}\label{claim:matrix_C_whole}
        The transformation $\pi\mapsto C^\pi$ forms a surjection from
        \[
            \Sym(\bool^{2n}) \twoheadrightarrow \mathscr{C}\;,
        \]
        where $\mathscr{C}$ is the set of all $2^n\times 2^n$ matrices over $\Zbb_{\geq 0}$ whose row sums and column sums are all equal to $2^n$.
    \end{restatable}
    
    \begin{restatable}{claim}{claimMatrixParts}\label{claim:matrix_C_parts}
        For every $G,H\in\Sym(\bool^n)^{\bool^n}$, the matrix $C^{F_{G,H}}$ coincides with the all-ones matrix $J$, i.e., $J_{L,U}:=1$ for all $L,U\in\bool^n$.
    \end{restatable}
    
    \paragraph{An auxiliary polynomial $f(r)$.}
    Combining \eqref{eq:invariance_rerandomize}, \eqref{eq:def_g_block_matrix}, and \cref{claim:matrix_C_parts}, we obtain
    \begin{align}
        \SD\lrp{\Acal^{F_{G,H}},\Acal^{\RP}}
        = \abs{\expn\sqbkt{g(C^{F_{G,H}})} - \expn\sqbkt{g(C^{\RP})} }
        = \abs{g(J) - \expn\sqbkt{g(C^{\RP})} }\;,\label{eq:SD_g_expected}
    \end{align}
    where the first equality uses \eqref{eq:invariance_rerandomize} together with $\Pr\sqbkt{\Acal^{\Rsf(\pi)}}=g(C^\pi)$ from \eqref{eq:def_g_block_matrix}, and the second uses \cref{claim:matrix_C_parts}. We recall from \eqref{eq:def_g_block_matrix} that $g$ is a polynomial of degree at most $2q$. Hence, with the monomial notation
    \[
        \mpow{C}{w}:=\prod_{L,U}C_{L,U}^{w_{L,U}}\;,
    \]
    we can write
    \[
        g(C)=\sum_{w}\widehat{g}(w)\,\mpow{C}{w}\;,
    \]
    where $w$ ranges over $\Zbb_{\geq 0}^{2^n\times 2^n}$ with $\abs{w}\leq 2q$ and $\widehat{g}(w)\in\Rbb$.
    
    We use the following notation for matrices $A,w\in\Zbb_{\geq 0}^{2^n\times 2^n}$ (with rows and columns indexed by $\bool^n$): $\rowsum{A}{L}:=\sum_U A_{L,U}$, $\colsum{A}{U}:=\sum_L A_{L,U}$, $\abs{A}:=\sum_{L,U}A_{L,U}$, and $A\leq w$ if $A_{L,U}\leq w_{L,U}$ for all $L,U$. We also write
    \[
        \stirling{w}{A}:=\prod_{L,U}\stirling{w_{L,U}}{A_{L,U}}\;.
    \]
    For real numbers $x,r\in\Rbb$ and an integer $k\geq 0$ we write $\gfall{x}{r}{k}:=\prod_{0\leq j<k}(x-jr)$, so that $\gfall{x}{1}{k}=(x)_k$ and $\gfall{x}{0}{k}=x^k$.
    Writing
    \begin{align*}
        \gamma_A(r)
        &:= \frac{1}{\gfall{2^{2n}}{r}{\abs{A}}}\;\prod_L\gfall{2^n}{r}{\rowsum{A}{L}}\prod_U\gfall{2^n}{r}{\colsum{A}{U}}\;,\\
        \tau_{w,A}(r) &:= \frac{\gfall{2^{2n}}{r}{2q}}{2^{4nq}} \widehat{g}(w) \stirling{w}{A} r^{\abs{w}-\abs{A}} \gamma_A(r) \;,
    \end{align*}
    we define the function $f$ as
    \begin{equation}
        f(r) := \sum_{w}\sum_{A\leq w}\tau_{w,A}(r)\;.
        \label{eq:domain_extension_def_f}
    \end{equation}
    We make the following claim about $f$, which we prove in \cref{sec:remaining_claim_proofs}.
    
    \begin{restatable}{claim}{claimDomainPolynomials}\label{claim:domain_polynomials}
        The function $f$ defined in \eqref{eq:domain_extension_def_f} is a polynomial of degree at most $4q$, and it satisfies
        \begin{equation}\label{eq:domain_extension_endpoints}
            f(0)=g(J)\;,\qquad
            f(1)=\prod_{0\leq j<2q}\lrp{1-\frac{j}{2^{2n}}}
            \cdot\expn\sqbkt{g(C^\RP)}\;,
        \end{equation}
        and $\abs{f(r)}\leq 2$ for every $r\in\set{0,\dots,\floor{2^n/(256q^2)}}$.
    \end{restatable}

    Assume $2^{2n}\geq 2q$, since otherwise $32770q^4/2^n>1$ and the lemma holds trivially. Then every factor $1-j/2^{2n}$ with $0\leq j<2q$ is positive, so we have
    \begin{align}
        \SD\lrp{\Acal^{F_{G,H}},\Acal^\RP}
        &= \abs{f(0) - \frac{f(1)}{\prod_{0\leq j<2q}\lrp{1-\frac{j}{2^{2n}}}}}\notag\\
        &\leq \abs{f(0) - f(1)} + \abs{f(1)}\cdot \lrp{\frac{1}{\prod_{0\leq j<2q}\lrp{1-\frac{j}{2^{2n}}}} - 1}\notag\\
        &= \abs{f(0) - f(1)} + \expn\sqbkt{g(C^{\RP})}\cdot\lrp{1-\prod_{0\leq j<2q}\lrp{1-\frac{j}{2^{2n}}}}\notag\\
        &\leq \abs{f(0) - f(1)} + \sum_{0\leq j<2q}\frac{j}{2^{2n}}
        \leq \abs{f(0) - f(1)} + \frac{2q^2}{2^{2n}}\;.\label{eq:SD_f0_f1}
    \end{align}
    The first inequality follows from triangle inequality, the second inequality follows from the fact that $0\leq \expn\sqbkt{g(C^\RP)}\leq 1$ together with $1-\prod_j(1-a_j)\leq\sum_j a_j$ for $a_j\in[0,1]$, the last inequality follows from $\sum_{0\leq j<2q}j=q(2q-1)\leq 2q^2$, and the two equalities follow from \cref{eq:domain_extension_endpoints}.

    \paragraph{Bounding $\abs{f(0) - f(1)}$.}
    To complete the proof, it suffices to bound $\abs{f(0)-f(1)}$, for which we use the Markov brothers' inequality twice. We will show that $\abs{f(0)-f(1)}\leq 32768q^4/2^n$ whenever $2^n\geq 16384q^4$. Let $T:=2^n/(256q^2)$, so that $T\geq 64q^2\geq 1$ under this assumption. By the mean value theorem and \eqref{eq:Markov_brother} on the interval $[0,T]\supseteq[0,1]$, we have
    \[
        \abs{f(0)-f(1)}
        \leq \max_{r\in[0,1]}\abs{f'(r)}
        \leq \frac{2(\deg f)^2}{T}\max_{r\in[0,T]}\abs{f(r)}
        \leq \frac{32q^2}{T}\max_{r\in[0,T]}\abs{f(r)}\;,
    \]
    where the last step uses $\deg f\leq 4q$ from \cref{claim:domain_polynomials}. To bound the remaining maximum, we apply \eqref{eq:polynomial_integers} to $f$ on the interval $[0,T]$, which is possible since $T\geq 64q^2\geq 4(\deg f)^2$. The integers in $[0,T]$ are exactly $r=0,\dots,\floor{T}$, and at each of them $\abs{f(r)}\leq 2$ by \cref{claim:domain_polynomials}. Therefore
    \[
        \max_{r\in[0,T]} \abs{f(r)}\leq 2\max_{r\in[0,T]\cap\Zbb}\abs{f(r)} \leq 4\;,
    \]
    and hence $\abs{f(0)-f(1)}\leq 128q^2/T=32768q^4/2^n$. Combining this with \eqref{eq:SD_f0_f1}, we get
    \[
        \SD\lrp{\Acal^{F_{G,H}},\Acal^{\RP}}
        \leq\abs{f(0)-f(1)} + \frac{2q^2}{2^{2n}}
        \leq \frac{32768q^4}{2^n} + \frac{2q^2}{2^{2n}}
        \leq \frac{32770q^4}{2^n}
    \]
    whenever $2^n\geq 16384q^4$. Otherwise, the right-hand side exceeds $1$, and the bound holds trivially. Plugging this into \eqref{eq:standard_hybrid_tweakable}, we conclude that
    \[
        \Adv^{\RP^{\leftrightarrow}}_{\PRP'}(\Acal)
        \leq \Adv^{\IC^{\leftrightarrow}}_{\TWK}(\Bcal)+ \frac{32770q^4}{2^n}\;,
    \]
    which completes the proof of \cref{lem:domain_extension}.
\end{proof}

\subsection{Remaining Proofs of the Claims}\label{sec:remaining_claim_proofs}

\claimMatrixWhole*
\begin{proof}[Proof~of~\cref{claim:matrix_C_whole}]
    We first show that $C^\pi\in\mathscr{C}$ for every $\pi\in\Sym(\bool^{2n})$. For every $U$, the column sum $\sum_L C^\pi_{L,U}$ is the number of pairs $(L,R)$ such that $\pi(L,R) \in \bool^n\times\set{U}$, which is $2^n$ since $\pi$ is a bijection. Likewise, for every $L$, the row sum $\sum_U C^\pi_{L,U}$ is the number of $R$ such that $\pi(L,R)\in\bool^{2n}$, which is $2^n$.

    We now show surjectivity: given any matrix $C\in\mathscr{C}$, we construct $\pi\in\Sym(\bool^{2n})$ such that $C^\pi = C$. First, for every $L\in\bool^n$, we pick a partition $\set{S_{L,U}}_{U\in\bool^n}$ of $\bool^n$ such that $\abs{S_{L,U}} = C_{L,U}$, which exists since the entries of row $L$ of $C$ are nonnegative integers summing to $2^n$. Likewise, for every $U\in\bool^n$, we pick a partition $\set{S'_{L,U}}_{L\in\bool^n}$ of $\bool^n$ such that $\abs{S'_{L,U}}=C_{L,U}$ for every $L$, which exists since column $U$ of $C$ sums to $2^n$. Next, for every $L,U\in\bool^n$ we pick any bijection
    \[
        \pi_{L,U}:\set{L}\times S_{L,U} \to S'_{L,U}\times\set{U}\;,
    \]
    which exists since $\abs{S_{L,U}} = C_{L,U} = \abs{S'_{L,U}}$. We then assemble these bijections into $\pi$ via
    \[
        \pi(L,R) := \pi_{L,U}(L,R)\;,
    \]
    where $U$ is the unique string such that $R\in S_{L,U}$; it is unique since $\set{S_{L,U}}_U$ is a partition of $\bool^n$. To see that $\pi$ is a bijection, note that for every $U$ the sets $\set{S'_{L,U}}_{L}$ form a partition of $\bool^n$, so every $(V,U)\in\bool^{2n}$ lies in exactly one of the sets $\Im(\pi_{L,U})=S'_{L,U}\times\set{U}$, and each $\pi_{L,U}$ is a bijection onto its image. Finally, by construction exactly $\abs{S_{L,U}}=C_{L,U}$ inputs from row $L$ are mapped into column $U$, and hence $C^\pi=C$.
\end{proof}

\claimMatrixParts*
\begin{proof}[Proof~of~\cref{claim:matrix_C_parts}]
    Fix $L,U\in\bool^n$. The second half of $F_{G,H}(L,R)$ equals $U$ exactly when $G(L,R)=U$. Since $G(L,\cdot)$ is a permutation, there is exactly one such $R$, namely $R=G^{-1}(L,U)$, where $G^{-1}(L,\cdot)$ denotes the inverse of $G(L,\cdot)$. Thus $C^{F_{G,H}}_{L,U}=1$ for every $L,U$, and hence $C^{F_{G,H}}=J$.
\end{proof}

\claimDomainPolynomials*
\begin{proof}[Proof~of~\cref{claim:domain_polynomials}]
    In addition to the notation introduced before \eqref{eq:domain_extension_def_f}, for matrices $A,B,w\in\Zbb_{\geq 0}^{2^n\times 2^n}$ and any matrix $C$ we write
    \[
        \binom{w}{B}:=\prod_{L,U}\binom{w_{L,U}}{B_{L,U}}\;,\qquad
        \mfall{C}{A}:=\prod_{L,U}(C_{L,U})_{A_{L,U}}\;.
    \]
    Since $\abs{A}\leq 2q$, each denominator $\gfall{2^{2n}}{r}{\abs{A}}$ in $\gamma_A(r)$ divides $\gfall{2^{2n}}{r}{2q}$, so each $\tau_{w,A}$ is a polynomial in $r$, and hence so is $f$. More precisely, using
    \[
        \gfall{2^{2n}}{r}{2q}=\gfall{2^{2n}}{r}{\abs{A}}\prod_{\abs{A}\leq j<2q}(2^{2n}-jr)\;,
    \]
    we can write
    \begin{equation}\label{eq:tau_polynomial}
        \tau_{w,A}(r)=\frac{\widehat{g}(w)}{2^{4nq}}\stirling{w}{A}\,r^{\abs{w}-\abs{A}}
        \prod_{L}\gfall{2^n}{r}{\rowsum{A}{L}}\prod_{U}\gfall{2^n}{r}{\colsum{A}{U}}
        \prod_{\abs{A}\leq j<2q}(2^{2n}-jr)\;.
    \end{equation}
    Since $\sum_L\rowsum{A}{L}=\sum_U\colsum{A}{U}=\abs{A}$, each $\tau_{w,A}$ has degree at most $(\abs{w}-\abs{A})+2\abs{A}+(2q-\abs{A})=2q+\abs{w}\leq 4q$ in $r$, so $\deg f\leq 4q$. Applying \eqref{eq:stirling} entrywise gives
    \ifarxivstyle
        \begin{align}
            \mpow{C}{w}=\prod_{L,U}\lrp{\sum_{A_{L,U}\leq w_{L,U}}\stirling{w_{L,U}}{A_{L,U}}(C_{L,U})_{A_{L,U}}}=\sum_{A\leq w}\lrp{\prod_{L,U}\stirling{w_{L,U}}{A_{L,U}}(C_{L,U})_{A_{L,U}}}=\sum_{A\leq w}\stirling{w}{A}\mfall{C}{A}\;.
            \label{eq:stirling_matrix}
        \end{align}
    \else
        \begin{align}
            \mpow{C}{w}&=\prod_{L,U}\lrp{\sum_{A_{L,U}\leq w_{L,U}}\stirling{w_{L,U}}{A_{L,U}}(C_{L,U})_{A_{L,U}}}\nonumber\\
            &=\sum_{A\leq w}\lrp{\prod_{L,U}\stirling{w_{L,U}}{A_{L,U}}(C_{L,U})_{A_{L,U}}}=\sum_{A\leq w}\stirling{w}{A}\mfall{C}{A}\;.
            \label{eq:stirling_matrix}
        \end{align}
    \fi

    \paragraph{The value $f(0)$.} With the convention $0^0=1$, $\tau_{w,A}(0)=0$ unless $\abs{A}=\abs{w}$, and together with $A\leq w$ this forces $A=w$. Since $\stirling{w}{w}=1$ and $\gfall{2^n}{0}{k}=2^{nk}$,
    \[
        \tau_{w,w}(0)=\frac{\widehat{g}(w)}{2^{4nq}}\cdot 2^{n\abs{w}}\cdot 2^{n\abs{w}}\cdot 2^{2n(2q-\abs{w})}=\widehat{g}(w)\;.
    \]
    Hence $f(0)=\sum_w\widehat{g}(w)=g(J)$, since every monomial $\mpow{C}{w}$ equals $1$ at $C=J$.

    \paragraph{The matrices $C_t$.} For an integer $t\geq 1$, let $\sigma_t:\bool^n\times[t]\to[t]\times\bool^n$ be a uniformly random bijection, and let $C_t$ be the $2^n\times 2^n$ matrix with entries
    \[
        C_{t,L,U}:=\abs{\set{i\in[t]\;\big|\;\sigma_t(L,i)\in[t]\times\set{U}}}\;.
    \]
    Every row and column of $C_t$ sums to $t$. Identifying $[2^n]$ with $\bool^n$, $\sigma_t$ is a uniform element of $\Sym(\bool^{2n})$ when $t=2^n$, so $C_{2^n}$ is distributed as $C^{\RP}$. We show that for every $A\in\Zbb_{\geq 0}^{2^n\times 2^n}$,
    \begin{equation}\label{eq:domain_extension_matrix_moments}
        (2^nt)_{\abs{A}}\cdot\expn\sqbkt{\mfall{C_t}{A}}
        =\prod_L(t)_{\rowsum{A}{L}}\prod_U(t)_{\colsum{A}{U}}\;.
    \end{equation}
    If $\rowsum{A}{L}>t$ for some $L$, both sides vanish: the right-hand side contains $(t)_{\rowsum{A}{L}}=0$, and since row $L$ of $C_t$ sums to $t<\rowsum{A}{L}$, some $A_{L,U}>C_{t,L,U}$, so $(C_{t,L,U})_{A_{L,U}}=0$. The same holds if $\colsum{A}{U}>t$ for some $U$. Otherwise $\abs{A}\leq 2^nt$. Write $\mathrm{In}_L:=\set{L}\times[t]$ and $\mathrm{Out}_U:=[t]\times\set{U}$. The factor $(C_{t,L,U})_{A_{L,U}}$ counts the ordered $A_{L,U}$-tuples of distinct elements of $\mathrm{In}_L$ mapped by $\sigma_t$ into $\mathrm{Out}_U$. Elements counted for different $(L,U)$ are distinct, since they lie in different $\mathrm{In}_L$ or are mapped into different $\mathrm{Out}_U$. Hence $\mfall{C_t}{A}$ is the number of pairs $(\mathbf{x},\mathbf{y})$ of families $\mathbf{x}=(x_{L,U,i})$ and $\mathbf{y}=(y_{L,U,i})$, indexed by $L,U$ and $i\leq A_{L,U}$, such that the $x_{L,U,i}\in\mathrm{In}_L$ are pairwise distinct, the $y_{L,U,i}\in\mathrm{Out}_U$ are pairwise distinct, and $\sigma_t(x_{L,U,i})=y_{L,U,i}$ for all indices. There are $\prod_L(t)_{\rowsum{A}{L}}$ choices of $\mathbf{x}$ and $\prod_U(t)_{\colsum{A}{U}}$ choices of $\mathbf{y}$, and for each such pair, the uniform bijection $\sigma_t$ on a set of size $2^nt$ satisfies the $\abs{A}$ prescribed equalities with probability $1/(2^nt)_{\abs{A}}$. Linearity of expectation gives \eqref{eq:domain_extension_matrix_moments}.

    \paragraph{The value $f(1)$.} Fix $w$ and $A\leq w$, so $\abs{A}\leq\abs{w}\leq 2q$. Using $\gfall{x}{1}{k}=(x)_k$, $(2^{2n})_{2q}=(2^{2n})_{\abs{A}}\prod_{\abs{A}\leq j<2q}(2^{2n}-j)$ and \eqref{eq:domain_extension_matrix_moments} with $t=2^n$,
    \ifarxivstyle
        \begin{align*}
            \tau_{w,A}(1)=\frac{\widehat{g}(w)}{2^{4nq}}\stirling{w}{A}\prod_L(2^n)_{\rowsum{A}{L}}\prod_U(2^n)_{\colsum{A}{U}}\prod_{\abs{A}\leq j<2q}(2^{2n}-j)=\frac{(2^{2n})_{2q}}{2^{4nq}}\,\widehat{g}(w)\stirling{w}{A}\,\expn\sqbkt{\mfall{C_{2^n}}{A}}\;.
        \end{align*}
    \else
        \begin{align*}
            \tau_{w,A}(1)
            &=\frac{\widehat{g}(w)}{2^{4nq}}\stirling{w}{A}\prod_L(2^n)_{\rowsum{A}{L}}\prod_U(2^n)_{\colsum{A}{U}}\prod_{\abs{A}\leq j<2q}(2^{2n}-j)\\
            &=\frac{(2^{2n})_{2q}}{2^{4nq}}\,\widehat{g}(w)\stirling{w}{A}\,\expn\sqbkt{\mfall{C_{2^n}}{A}}\;.
        \end{align*}
    \fi
    We have $(2^{2n})_{2q}/2^{4nq}=\prod_{0\leq j<2q}(1-j/2^{2n})$. Summing over $A\leq w$ using \eqref{eq:stirling_matrix}, then over $w$, yields
    \[
        f(1)=\prod_{0\leq j<2q}\lrp{1-\frac{j}{2^{2n}}}\cdot\expn\sqbkt{g(C_{2^n})}
        =\prod_{0\leq j<2q}\lrp{1-\frac{j}{2^{2n}}}\cdot\expn\sqbkt{g(C^{\RP})}\;.
    \]

    \paragraph{Bounding $\abs{f(r)}$: the padded family $f_r$.} For $r=0$, $\abs{f(0)}=g(J)\leq 1$, since $J=C^{F_{G,H}}$ by \cref{claim:matrix_C_parts} and $g(C^\pi)=\Pr\sqbkt{\Acal^{\Rsf(\pi)}}\in[0,1]$ by \eqref{eq:def_g_block_matrix}. Now fix an integer $1\leq r\leq 2^n/(256q^2)$, and put $t_*:=2^n/r\geq 256q^2$, which need not be an integer. For integers $s\geq 0$, let $W_s$ be the sum of $s$ independent uniformly random $2^n\times 2^n$ permutation matrices ($W_0:=0$). For every integer $t\in[1,t_*]$, sample $C_t$ and $W_{2^n-rt}$ independently and define
    \begin{equation}\label{eq:domain_extension_def_f_r}
        f_r(t):=\frac{(2^nt)_{2q}}{(2^nt_*)_{2q}}\cdot\expn\sqbkt{g(rC_t+W_{2^n-rt})}\;.
    \end{equation}
    The matrix $rC_t+W_{2^n-rt}$ has nonnegative integer entries and all row and column sums equal to $rt+(2^n-rt)=2^n$, so by \cref{claim:matrix_C_whole} it equals $C^\pi$ for some $\pi$, and $g$ takes a value in $[0,1]$ on it. Since $2q\leq 2^n\leq 2^nt\leq 2^nt_*$, the ratio in \eqref{eq:domain_extension_def_f_r} lies in $(0,1]$. Hence $0\leq f_r(t)\leq 1$ for every integer $t\in[1,t_*]$.

    \paragraph{Moments of $W_s$.} Fix $B\in\Zbb_{\geq 0}^{2^n\times 2^n}$ and list the positions $(L,U)$ with multiplicity $B_{L,U}$ as $(L_1,U_1),\dots,(L_k,U_k)$, where $k=\abs{B}$; that is, each position $(L,U)$ appears exactly $B_{L,U}$ times in the list. Write $W_s=\sum_{\ell=1}^sM_\ell$ with $M_1,\dots,M_s$ independent uniform permutation matrices, and let $M$ denote a single uniform permutation matrix. As for $C_t$, we write $W_{s,L,U}$ and $M_{\ell,L,U}$ for the $(L,U)$-th entries of $W_s$ and $M_\ell$. By the choice of $(L_1,U_1),\dots,(L_k,U_k)$, we have
    \[
        \mpow{W_s}{B}=\prod_{L,U}W_{s,L,U}^{B_{L,U}}=\prod_{m=1}^{k}W_{s,L_m,U_m}=\prod_{m=1}^{k}\lrp{\sum_{\ell=1}^{s}M_{\ell,L_m,U_m}}.
    \]
    When we expand this product, each summand is a product of $k$ entries of the $M_\ell$, one from each factor, for a total of $s^k$ summands. With each summand we can associate a function $\varphi:[k]\to[s]$ by setting $\varphi(m)$ to be the index $\ell$ such that $M_{\ell,L_m,U_m}$ is the $m$-th factor in the product. With this notation we can write
    \[
        \mpow{W_s}{B}=\sum_{\varphi:[k]\to[s]}\prod_{m=1}^{k}M_{\varphi(m),L_m,U_m}\;.
    \]
    By linearity of expectation, we obtain
    \begin{equation}\label{eq:Ws_moment_expansion}
        \expn\sqbkt{\mpow{W_s}{B}}
        =\sum_{\varphi:[k]\to[s]}\expn\sqbkt{\prod_{m=1}^{k}M_{\varphi(m),L_m,U_m}}\;.
    \end{equation}

    \paragraph{The polynomial $\mu_B$.} We will show that $\expn\sqbkt{\mpow{W_s}{B}}$ is a polynomial in $s$ of degree at most $k$. Fix $\varphi:[k]\to[s]$, and let $\ell_1<\dots<\ell_i$ be the distinct values in the range of $\varphi$. Define $K_b:=\varphi^{-1}(\ell_b)$ for $b\in[i]$, so that $\kappa(\varphi):=\set{K_1,\dots,K_i}$ is a partition of $[k]$. Then we can write
    \[
        \prod_{m=1}^{k}M_{\varphi(m),L_m,U_m}=\prod_{b=1}^{i}\;\prod_{m\in K_b}M_{\ell_b,L_m,U_m}.
    \]
    The $b$th inner product $\prod_{m\in K_b}M_{\ell_b,L_m,U_m}$ depends only on the entries of $M_{\ell_b}$, and $M_{\ell_1},\dots,M_{\ell_i}$ are independent since the $\ell_b$ are distinct, so the expectation of the product factorizes. Each $M_{\ell_b}$ is distributed as $M$, hence
    \[
        \expn\sqbkt{\prod_{m=1}^{k}M_{\varphi(m),L_m,U_m}}
        =\prod_{b=1}^{i}\expn\sqbkt{\prod_{m\in K_b}M_{L_m,U_m}}
        =:c_{\kappa(\varphi)}\;.
    \]
    We observe that $c_{\kappa(\varphi)}$ depends on $\varphi$ only through the unordered partition $\kappa(\varphi)$.
    
    Now fix a partition $\kappa$ of $[k]$ into nonempty blocks $K_1,\dots,K_i$, indexed so that $\min K_1<\dots<\min K_i$. A map $\varphi:[k]\to[s]$ has $\kappa(\varphi)=\kappa$ if and only if $\varphi$ is constant on each block $K_b$ and takes distinct values on distinct blocks. Thus the number of such maps is exactly the number of ways to choose a sequence of $i$ distinct elements from $[s]$, which is $(s)_i$.

    Grouping the terms of \eqref{eq:Ws_moment_expansion} by $\kappa(\varphi)$ gives
    \ifarxivstyle
        \begin{align}
            \expn\sqbkt{\mpow{W_s}{B}}=\sum_{\varphi:[k]\to[s]}c_{\kappa(\varphi)}=\sum_{\kappa}\sum_{\varphi:\kappa(\varphi)=\kappa}c_\kappa=\sum_{\kappa}c_\kappa\abs{\set{\varphi:\kappa(\varphi)=\kappa}}
            =\sum_{\kappa}c_\kappa(s)_{\abs{\kappa}} =: \mu_B(s)\;.
            \label{eq:Ws_moment_polynomial}
        \end{align}
    \else
        \begin{align}
            \expn\sqbkt{\mpow{W_s}{B}}
            &=\sum_{\varphi:[k]\to[s]}c_{\kappa(\varphi)} \nonumber\\
            &=\sum_{\kappa}\sum_{\varphi:\kappa(\varphi)=\kappa}c_\kappa \nonumber\\
            &=\sum_{\kappa}c_\kappa\abs{\set{\varphi:\kappa(\varphi)=\kappa}}
            =\sum_{\kappa}c_\kappa(s)_{\abs{\kappa}} =: \mu_B(s)\;.
            \label{eq:Ws_moment_polynomial}
        \end{align}
    \fi
    Since $(s)_i$ is a polynomial in $s$ of degree $i$ and every partition of $[k]$ has at most $k$ blocks, $\mu_B$ is a polynomial of degree at most $k = \abs{B}$. If $B=0$, $k=0$, and the only partition of $[0]=\emptyset$ is the empty one, with $\abs{\kappa}=0$ and $c_\kappa=1$ (an empty product), so $\mu_0=(s)_0=1$. If $B\neq 0$, then $k\geq 1$, so every partition of $[k]$ has $\abs{\kappa}\geq 1$ blocks, and $(0)_i=0$ for $i\geq 1$; hence $\mu_B(0)=0$.

    \paragraph{A polynomial $P_r$ extending $f_r$.} We will show that $f_r$ agrees at every integer $t\in[1,t_*]$ with a polynomial $P_r$ in $t$, which, unlike $f_r$, can also be evaluated at the possibly non-integer point $t_*$. Fix $w$ with $\abs{w}\leq 2q$ and an integer $t\in[1,t_*]$, and let $s:=2^n-rt$. Then $s$ is an integer with $s\geq 2^n-rt_*=0$, so \eqref{eq:Ws_moment_polynomial} applies to $W_s$. Applying the binomial theorem to each entry of $rC_t+W_s$, we have
    \ifarxivstyle
        \begin{align*}
            \mpow{(rC_t+W_s)}{w}=\prod_{L,U}\lrp{rC_{t,L,U}+W_{s,L,U}}^{w_{L,U}}
            &=\prod_{L,U}\lrp{\sum_{B_{L,U}\leq w_{L,U}}\binom{w_{L,U}}{B_{L,U}}r^{w_{L,U}-B_{L,U}}C_{t,L,U}^{w_{L,U}-B_{L,U}}W_{s,L,U}^{B_{L,U}}}\\
            &=\sum_{B\leq w}\binom{w}{B}r^{\abs{w}-\abs{B}}\,\mpow{C_t}{(w-B)}\,\mpow{W_s}{B}\;.
        \end{align*}
    \else
        \begin{align*}
            \mpow{(rC_t+W_s)}{w}
            &=\prod_{L,U}\lrp{rC_{t,L,U}+W_{s,L,U}}^{w_{L,U}}\\
            &=\prod_{L,U}\lrp{\sum_{B_{L,U}\leq w_{L,U}}\binom{w_{L,U}}{B_{L,U}}r^{w_{L,U}-B_{L,U}}C_{t,L,U}^{w_{L,U}-B_{L,U}}W_{s,L,U}^{B_{L,U}}}\\
            &=\sum_{B\leq w}\binom{w}{B}r^{\abs{w}-\abs{B}}\,\mpow{C_t}{(w-B)}\,\mpow{W_s}{B}\;.
        \end{align*}
    \fi
    In the last step, each term of the expanded product picks one summand $B_{L,U}$ from the factor at every position $(L,U)$, and we collect these choices into a matrix $B\leq w$; the powers of $r$ then multiply to $r^{\sum_{L,U}(w_{L,U}-B_{L,U})}=r^{\abs{w}-\abs{B}}$.

    Next we take expectations. Since $C_t$ and $W_s$ are independent, the expectation of each term factorizes as $\expn\sqbkt{\mpow{C_t}{(w-B)}}\cdot\expn\sqbkt{\mpow{W_s}{B}}$, and the second factor equals $\mu_B(s)=\mu_B(2^n-rt)$ by \eqref{eq:Ws_moment_polynomial}. For the first factor, we apply \eqref{eq:stirling_matrix} to $C_t$ with $w-B$ in place of $w$, and obtain $\expn\sqbkt{\mpow{C_t}{(w-B)}}=\sum_{A\leq w-B}\stirling{w-B}{A}\expn\sqbkt{\mfall{C_t}{A}}$ by linearity of expectation. Altogether,
    \begin{equation}\label{eq:Pr_expectation}
        \expn\sqbkt{\mpow{(rC_t+W_s)}{w}}
        =\sum_{B\leq w}\binom{w}{B}r^{\abs{w}-\abs{B}}\mu_B(2^n-rt)\sum_{A\leq w-B}\stirling{w-B}{A}\expn\sqbkt{\mfall{C_t}{A}}\;.
    \end{equation}

    We now clear the denominators hidden in the moments $\expn\sqbkt{\mfall{C_t}{A}}$. Every $A$ in the inner sum of \eqref{eq:Pr_expectation} satisfies $\abs{A}\leq\abs{w-B}\leq\abs{w}\leq 2q$, so we can split the product defining $(2^nt)_{2q}$ at $j=\abs{A}$, and then apply \eqref{eq:domain_extension_matrix_moments}:
    \begin{align*}
        (2^nt)_{2q}\cdot\expn\sqbkt{\mfall{C_t}{A}}
        &=\prod_{\abs{A}\leq j<2q}(2^nt-j)\cdot(2^nt)_{\abs{A}}\expn\sqbkt{\mfall{C_t}{A}}\\
        &=\prod_L(t)_{\rowsum{A}{L}}\prod_U(t)_{\colsum{A}{U}}\prod_{\abs{A}\leq j<2q}(2^nt-j)
        =:\rho_A(t)\;.
    \end{align*}
    Since $g(C)=\sum_w\widehat{g}(w)\mpow{C}{w}$ is a finite sum, linearity of expectation gives
    \[
        \expn\sqbkt{g(rC_t+W_s)}=\sum_w\widehat{g}(w)\expn\sqbkt{\mpow{(rC_t+W_s)}{w}}\;.
    \]
    Plugging this and \eqref{eq:Pr_expectation} into \eqref{eq:domain_extension_def_f_r}, and multiplying each moment $\expn\sqbkt{\mfall{C_t}{A}}$ by the factor $(2^nt)_{2q}$ from \eqref{eq:domain_extension_def_f_r}, we obtain $f_r(t)=P_r(t)$ for every integer $t\in[1,t_*]$, where
    \begin{equation}\label{eq:Pr}
        P_r(t):=\frac{1}{(2^nt_*)_{2q}}\sum_w\widehat{g}(w)\sum_{B\leq w}\binom{w}{B}r^{\abs{w}-\abs{B}}\mu_B(2^n-rt)\sum_{A\leq w-B}\stirling{w-B}{A}\rho_A(t)\;.
    \end{equation}
    The right-hand side of \eqref{eq:Pr} is defined for every real $t$, and it is a polynomial in $t$. Indeed, the prefactor does not depend on $t$, and it is well defined since each factor of $(2^nt_*)_{2q}$ satisfies $2^nt_*-j=2^{2n}/r-j\geq 256q^2\cdot 2^n-2q>0$ for $j<2q$. Moreover, $\mu_B(2^n-rt)$ is the polynomial $\mu_B$ composed with the affine map $t\mapsto 2^n-rt$, $\rho_A(t)$ is a product of linear polynomials in $t$, and all sums in \eqref{eq:Pr} are finite.

    \paragraph{Degree of $P_r$.} We will show that $\deg P_r\leq 4q$, by bounding the degree in $t$ of each $(w,B,A)$ term of \eqref{eq:Pr}. The factors $\widehat{g}(w)$, $\binom{w}{B}$, $r^{\abs{w}-\abs{B}}$, $\stirling{w-B}{A}$ and the prefactor do not depend on $t$. Composing a polynomial with an affine map does not increase its degree, so $\mu_B(2^n-rt)$ has degree at most $\deg\mu_B\leq\abs{B}$ in $t$. The polynomial $\rho_A(t)$ is a product of linear factors in $t$: each $(t)_k$ contributes $k$ of them, and the last product in $\rho_A(t)$ contributes $2q-\abs{A}$ of them. Since $\sum_L\rowsum{A}{L}=\sum_U\colsum{A}{U}=\abs{A}$, we get
    \[
        \deg\rho_A
        =\sum_L\rowsum{A}{L}+\sum_U\colsum{A}{U}+(2q-\abs{A})
        =\abs{A}+\abs{A}+(2q-\abs{A})
        =2q+\abs{A}\;.
    \]
    Finally, $A\leq w-B$ implies $\abs{A}\leq\abs{w}-\abs{B}$, so the degree of the $(w,B,A)$ term is at most
    \[
        \abs{B}+2q+\abs{A}
        \leq\abs{B}+2q+\abs{w}-\abs{B}
        =2q+\abs{w}
        \leq 4q\;,
    \]
    where the last step uses $\abs{w}\leq 2q$. Since $P_r$ is a finite sum of such terms, $\deg P_r\leq 4q$.

    \paragraph{The value $P_r(t_*)$.} We will show that
    \begin{equation}\label{eq:f-vs-Pr}
        f(r)=\prod_{0\leq j<2q}\lrp{1-\frac{jr}{2^{2n}}}\cdot P_r(t_*)\;.
    \end{equation}
    At $t=t_*$, the argument of $\mu_B$ in \eqref{eq:Pr} is $2^n-rt_*=0$, and we have shown that $\mu_B(0)=0$ for $B\neq 0$ and $\mu_0=1$. Hence only the terms with $B=0$ remain, and since $\binom{w}{0}=1$ and $w-0=w$,
    \begin{equation}\label{eq:Pr_at_tstar}
        P_r(t_*)=\sum_w\sum_{A\leq w}\widehat{g}(w)\stirling{w}{A}\cdot\frac{r^{\abs{w}}\rho_A(t_*)}{(2^nt_*)_{2q}}\;.
    \end{equation}
    To compare this with $f(r)$, we rewrite the last factor in terms of generalized falling factorials. For every integer $k\geq 0$ and real $x$, we have
    \[
        (x/r)_k=\prod_{0\leq j<k}\lrp{\frac{x}{r}-j}=\prod_{0\leq j<k}\frac{x-jr}{r}=r^{-k}\gfall{x}{r}{k}\;.
    \]
    Recall that $t_*=2^n/r$ and $2^nt_*=2^{2n}/r$. We rewrite the three products in $\rho_A(t_*)$ one by one. Applying the identity with $x=2^n$ to each $(t_*)_{\rowsum{A}{L}}$ and each $(t_*)_{\colsum{A}{U}}$, and using $\sum_L\rowsum{A}{L}=\sum_U\colsum{A}{U}=\abs{A}$ to collect the powers of $r$, we get
    \[
        \prod_L(t_*)_{\rowsum{A}{L}}=r^{-\abs{A}}\prod_L\gfall{2^n}{r}{\rowsum{A}{L}}\;,\qquad
        \prod_U(t_*)_{\colsum{A}{U}}=r^{-\abs{A}}\prod_U\gfall{2^n}{r}{\colsum{A}{U}}\;.
    \]
    Each of the $2q-\abs{A}$ factors of the last product in $\rho_A(t_*)$ satisfies $2^nt_*-j=2^{2n}/r-j=r^{-1}(2^{2n}-jr)$, so
    \[
        \prod_{\abs{A}\leq j<2q}(2^nt_*-j)=r^{-(2q-\abs{A})}\prod_{\abs{A}\leq j<2q}(2^{2n}-jr)\;.
    \]
    Multiplying these three expressions, and noting that the exponents of $r$ add up to $-\abs{A}-\abs{A}-(2q-\abs{A})=-(2q+\abs{A})$, we get
    \[
        \rho_A(t_*)=r^{-(2q+\abs{A})}\prod_L\gfall{2^n}{r}{\rowsum{A}{L}}\prod_U\gfall{2^n}{r}{\colsum{A}{U}}\prod_{\abs{A}\leq j<2q}(2^{2n}-jr)\;.
    \]
    For the denominator, applying the identity with $x=2^{2n}$ and $k=2q$ gives
    \[
        (2^nt_*)_{2q}=(2^{2n}/r)_{2q}=r^{-2q}\gfall{2^{2n}}{r}{2q}\;,
    \]
    which is nonzero, since $jr<2q\cdot 2^n/(256q^2)<2^{2n}$ for every $j<2q$. Dividing $r^{\abs{w}}\rho_A(t_*)$ by this, the powers of $r$ combine to $r^{\abs{w}-(2q+\abs{A})+2q}=r^{\abs{w}-\abs{A}}$, so
    \ifarxivstyle
        \begin{align*}
            \frac{r^{\abs{w}}\rho_A(t_*)}{(2^nt_*)_{2q}}
            &=\frac{r^{\abs{w}-\abs{A}}}{\gfall{2^{2n}}{r}{2q}}\prod_L\gfall{2^n}{r}{\rowsum{A}{L}}\prod_U\gfall{2^n}{r}{\colsum{A}{U}}\prod_{\abs{A}\leq j<2q}(2^{2n}-jr)\\
            &=\frac{2^{4nq}}{\gfall{2^{2n}}{r}{2q}}\cdot\frac{r^{\abs{w}-\abs{A}}}{2^{4nq}}\prod_L\gfall{2^n}{r}{\rowsum{A}{L}}\cdot\prod_U\gfall{2^n}{r}{\colsum{A}{U}}\prod_{\abs{A}\leq j<2q}(2^{2n}-jr)\;.
        \end{align*}
    \else
        \begin{align*}
            \frac{r^{\abs{w}}\rho_A(t_*)}{(2^nt_*)_{2q}}
            &=\frac{r^{\abs{w}-\abs{A}}}{\gfall{2^{2n}}{r}{2q}}\prod_L\gfall{2^n}{r}{\rowsum{A}{L}}\prod_U\gfall{2^n}{r}{\colsum{A}{U}}\prod_{\abs{A}\leq j<2q}(2^{2n}-jr)\\
            &=\frac{2^{4nq}}{\gfall{2^{2n}}{r}{2q}}\cdot\frac{r^{\abs{w}-\abs{A}}}{2^{4nq}}\prod_L\gfall{2^n}{r}{\rowsum{A}{L}}\\
            &\qquad\cdot\prod_U\gfall{2^n}{r}{\colsum{A}{U}}\prod_{\abs{A}\leq j<2q}(2^{2n}-jr)\;.
        \end{align*}
    \fi
    Multiplying by $\widehat{g}(w)\stirling{w}{A}$ and comparing the last line with \eqref{eq:tau_polynomial}, we see that
    \[
        \widehat{g}(w)\stirling{w}{A}\cdot\frac{r^{\abs{w}}\rho_A(t_*)}{(2^nt_*)_{2q}}
        =\frac{2^{4nq}}{\gfall{2^{2n}}{r}{2q}}\cdot\tau_{w,A}(r)\;.
    \]
    Substituting this into \eqref{eq:Pr_at_tstar} and using $f=\sum_w\sum_{A\leq w}\tau_{w,A}$, we get
    \[
        P_r(t_*)
        =\frac{2^{4nq}}{\gfall{2^{2n}}{r}{2q}}\sum_w\sum_{A\leq w}\tau_{w,A}(r)
        =\frac{2^{4nq}}{\gfall{2^{2n}}{r}{2q}}\cdot f(r)\;.
    \]
    Finally, since $2^{4nq}=(2^{2n})^{2q}$, we have
    \[
        \frac{\gfall{2^{2n}}{r}{2q}}{2^{4nq}}
        =\prod_{0\leq j<2q}\frac{2^{2n}-jr}{2^{2n}}
        =\prod_{0\leq j<2q}\lrp{1-\frac{jr}{2^{2n}}}\;,
    \]
    so multiplying both sides of the previous display by this product gives \eqref{eq:f-vs-Pr}. Moreover, since $0\leq jr<2^{2n}$ for every $j<2q$, each factor of this product lies in $(0,1]$, and hence so does the product.

    \paragraph{Concluding the bound.} We now apply \eqref{eq:polynomial_integers} to the polynomial $P_r$ on the interval $[a,b]:=[1,t_*]$. This requires $b-a\geq 4(\deg P_r)^2$. Since $\deg P_r\leq 4q$, it suffices that $t_*-1\geq 64q^2$, which holds because $t_*\geq 256q^2$ and $q\geq 1$. The integers in $[1,t_*]$ are exactly $t=1,\dots,\floor{t_*}$, and at each of them we have $P_r(t)=f_r(t)\in[0,1]$. Therefore \eqref{eq:polynomial_integers} gives
    \[
        \abs{P_r(t_*)}
        \leq 2\max_{t\in[1,t_*]\cap\Zbb}\abs{P_r(t)}
        =2\max_{t\in[1,t_*]\cap\Zbb}\abs{f_r(t)}
        \leq 2\;.
    \]
    Combining this with \eqref{eq:f-vs-Pr}, and using that the product in \eqref{eq:f-vs-Pr} lies in $(0,1]$, we obtain
    \[
        \abs{f(r)}=\prod_{0\leq j<2q}\lrp{1-\frac{jr}{2^{2n}}}\cdot\abs{P_r(t_*)}\leq\abs{P_r(t_*)}\leq 2\;.
    \]
    This holds for every integer $1\leq r\leq 2^n/(256q^2)$. Together with $\abs{f(0)}\leq 1$, which we showed at the beginning of the paragraph on the padded family, we conclude that $\abs{f(r)}\leq 2$ for every $r\in\set{0,\dots,\floor{2^n/(256q^2)}}$, which completes the proof of \cref{claim:domain_polynomials}.
\end{proof}

\phantom{\cite{}}
\bibliographystyle{alpha}
\bibliography{cryptobib/abbrev3,cryptobib/crypto,ref.bib}

\appendix

\section{AI Disclosure}
We use AI in several ways, including generating TikZ code, identifying calculation, spelling, and grammatical errors, and discussing the presentation of the paper.
We also use AI for interactive technical discussions, which are intertwined with our own ideas and discussions. Our main result (\cref{thm:ind_lift}) has gone through multiple iterations of technical improvements. One of these arose from an insight obtained during a discussion with GPT-5.6~Sol~Ultra: using an existential argument based on the polynomial method, one can unnoticeably embed a base oracle into a keyed family of its instances.
The proof of \cref{lem:domain_extension} has also undergone a similar iterative process, in which GPT-6~Astra helped us with several counting arguments concerning randomized matrices and contributed to the idea of rerandomizing the permutation oracle, though this idea was buried in a single sentence among several pages of derivations.

\section{Efficient Simulation of Biased Random Oracles}\label{sec:efficient_biased_RO_simulation}

In \cref{sec:lifting}, our reduction samples a parameter $p$ according to a certain distribution $\Dcal$ over the interval $[0,1]$, and then implements an oracle $H:\Kcal\to\bool$ in which every entry $H(k)\gets\Ber(p)$ is sampled independently according to the biased Bernoulli distribution of parameter $p$. We discuss here how this can be efficiently simulated, in the following two different cases.

\subsubsection{Case 1: $\Dcal$ is efficiently sampleable.} In the case when $\Dcal$ is efficiently sampleable, one can again follow an approach similar to the compressed oracle technique \cite{C:Zhandry19,EC:CFHL21}.
Once the parameter $p$ is sampled and fixed, we can simply use the Bernoulli recording oracle by \cite{CGKSW23}. On a very high level, consider the purification of $H$ stored in a designated register $\Dsf$, which is initially of the form $\bigotimes_{k\in\Kcal}\ket{\hat{0}_p}$, with
$$
\ket{\hat{0}_p} := \sqrt{1-p} \ket{0} + \sqrt{p}\ket{1}\;,\qquad \ket{\hat{1}_p} := \sqrt{p}\ket{0} - \sqrt{1-p}\ket{1}\;.
$$
We can then compress it into a compact representation via applying the compression operator $\Csf:=\ket{\bot}\bra{\hat{0}_p} + \ket{\hat{0}_p}\bra{\bot} + \ket{\hat{1}_p}\bra{\hat{1}_p}$ to each entry of $\Dsf$, which would compress the initial state to an empty database $\ket{\bot}$. In each of the queries, one would simply need to implement the unitary $\Csf\cdot \Osf\cdot\Csf$ for $\Osf: \ket{k,y,H}\mapsto\ket{k,y\oplus H(k),H}$, which can indeed be made efficient.

\subsubsection{Case 2: $\Dcal$ is uniform over $[0,1]$.} In the case when $\Dcal$ is uniform over $[0,1]$, the resulting parameter $p\in[0,1]$ would have infinitely large entropy, and so it is not even possible to store the entire $p$. Nevertheless, there are two possible ways out of this problem.

First, if we simply pick a sufficiently high precision $\lambda\in\Zbb_{>0}$ and truncate $p$ to a $\lambda$-bit approximation $p_\lambda := \lfloor p\cdot 2^\lambda\rceil\cdot 2^{-\lambda}$, then given that $\abs{p_{\lambda}-p}\leq 2^{-\lambda}$, the resulting oracle $H_\lambda \gets\Ber(p_\lambda)^{\otimes\Kcal}$ would differ by a statistical distance at most $\SD(H_\lambda,H)\leq \abs{\Kcal}\cdot 2^{-\lambda}$. Since $p_\lambda$ is efficiently sampleable, we have fallen back to the first case, in which we know how to simulate $H_\lambda$.

On the other hand, if we insist on a perfect simulation (without any loss), the simulation can also be done without sampling the entire $p$. Observe that for any oracle algorithm $\Acal$ making at most $q$ quantum queries to $H$, conditioned on any fixed choice of $H$, the probability of $\Acal$ outputting $1$ is a multilinear polynomial $a(H)$ of degree at most $2q$ in $(H(k))_k$. Therefore, it suffices to simulate a randomized function $F(k)$ such that for every distinct $k_1,\dots,k_\ell$ with $\ell\leq 2q$, we have
\begin{equation}
    \expn_F\sqbkt{\prod_{i\in[\ell]} F(k_i)} = \expn_H\sqbkt{\prod_{i\in[\ell]} H(k_i)} = \int_0^1 p^\ell\,dp = \frac{1}{\ell +1}\;,\label{eq:biased_2q_wise_independent_hash}
\end{equation}
where the last two equalities follow from straight-forward calculation.
Then we would have
\[
\Pr\sqbkt{\Acal^F}
= \expn_F\sqbkt{a(F)}
= \sum_{S\subseteq \Kcal:\abs{S}\leq 2q} \frac{\widehat{a}(S)}{\abs{S} +1}
= \expn_H\sqbkt{a(H)}
= \Pr\sqbkt{\Acal^H}\;,
\]
where $\widehat{a}(S)$ is the coefficient of the monomial $\prod_{k\in S}H(k)$ in $a(H)$.

To achieve \eqref{eq:biased_2q_wise_independent_hash}, we solve for $w_0,\dots,w_M\in \Qbb_{\geq 0}$ such that
\[
\sum_{0\leq j\leq M} w_j \cdot(j/M)^\ell = \frac{1}{\ell +1}\quad\forall \ell\in \set{0,\dots, 2q}
\]
with the convention that $0^0 := 1$,
since then sampling $p'$ with the distribution such that $\Pr\sqbkt{p'=j/M} := w_j$, and then sampling $F\gets \Ber(p')^{\otimes \Kcal}$ would immediately achieve \eqref{eq:biased_2q_wise_independent_hash}. Observe that this is a numerical integration formula in disguise: namely, for $f_\ell(x) := x^\ell$, we are given its evaluation $f_\ell(j/M)$ for each $j\in \set{0,\dots, M}$, and tasked to find a set of coefficients $w_0,\dots, w_M$ that recovers $\int_0^1 f_\ell(x)\,dx = 1/(\ell+1)$ via the formula $\sum_{0\leq j\leq M} w_j f_\ell(j/M)$ for all $\ell\in[2q]$ simultaneously. By Equation (5.5) in \cite{DFPW19}, such a formula would indeed exist for some $M\leq \max_{\ell}(\deg f_\ell +1)^2 = (2q+1)^2$. Furthermore, as has been pointed out in the same paper, the coefficients $w_0,\dots,w_M$ can be cast as a solution of a linear program, which can be solved efficiently.

\section{Lifting Hardness of Search Games}\label{sec:search_lifting}
In this section, we show that our indistinguishability lifting theorem (\cref{thm:ind_lift}) is useful for lifting hardness of search games. Specifically, it can be used to recover a weaker form of the measure-and-reprogram technique \cite{C:DFMS19,C:DonFehMaj20}. Such a weaker form can be seen in the literature. For example, Falcon uses a multi-target SIS assumption \cite{EC:FGGJK26,EPRINT:BacMaj26}, which can also be captured by our statement here.

\begin{corollary}
    Let $V:\Ycal\times\Zcal\to\bool$ be a predicate that can be evaluated efficiently, and $H:\Xcal\to\Ycal$ be a random oracle. Then, for every oracle algorithm $\Acal$ making at most $q$ quantum queries to $H$, there exists a similarly efficient algorithm $\Bcal$ (where we count each evaluation of $V$ as unit runtime), such that
    $$
    \Pr\left[V(H(X),Z)\;\Big|\;\begin{gathered}
       (X,Z)\gets\Acal^H\\ 
    \end{gathered}\right]
    \leq 8(q+1)^2\cdot\Pr\left[V(Y,Z)\;\middle|\;\begin{gathered}
        Y\gets \Ycal\\
        Z\gets\Bcal(Y)
    \end{gathered}\right]\;.
    $$
\end{corollary}
\begin{proof}
    For $b\in\bool$, consider the oracle $\Osf_b$ defined with two modes, with the selected mode indicated in a designated input register that is always measured by the oracle at the beginning of each query.
    On the read mode, the oracle always outputs $Y$ that is sampled uniformly at the initialization stage. On the challenge mode, the oracle takes an input $z$, fully measures it and records it in an internal state, and then outputs $b\land V(Y,z)$. We summarize this oracle in \cref{fig:O_search_lifting}.

    Furthermore, let $\Osf_{b}^\Xcal$ be the corresponding keyed oracle, which can be efficiently simulated, since querying the read mode is indistinguishable from querying a random oracle, while querying a challenge mode can be implemented coherently because $V$ can be evaluated efficiently.
    \begin{figure}[H]
        \centering
        \begin{minipage}[t]{\ifarxivstyle%
            0.25\linewidth%
            \else%
            0.4\linewidth%
            \fi}
            $\Osf_b(\text{``read''}):= Y\stackrel{\text{init}}{\gets}\Ycal$
        \end{minipage}
        \begin{minipage}[t]{\ifarxivstyle%
            0.25\linewidth%
            \else%
            0.4\linewidth%
            \fi}
            $\Osf_b(\text{``challenge''}, z)$:\vspace{-1em}
            \begin{algorithmic}[1]
                \STATE measure and record $z$
                \STATE {\bf return} $b\land V(Y,z)$
            \end{algorithmic}
        \end{minipage}
        \caption{The oracle $\Osf_b$ used for lifting hardness of search games.}\label{fig:O_search_lifting}
    \end{figure}

    Next, we transform the search algorithm $\Acal$ querying $H$ to a distinguishing algorithm $\Dcal$ querying $\Osf_{b}^\Xcal$ as follows: the algorithm $\Dcal$ runs $\Acal$, and forwards each query to $H$ to $\Osf_b^\Xcal$ in the read mode. After obtaining the output $(X,Z)$ of $\Acal$, it makes one more query $B\gets\Osf_b^\Xcal(X,\text{``challenge''},Z)$ and outputs the response $B$.
    By construction we then have
    \begin{align*}
        \Pr\left[
            \Dcal^{\Osf_{b}^\Xcal}
        \right]
        = b\cdot\epsilon\;,\text{ where}\quad
        \epsilon := \Pr\Big[
            \begin{gathered}V(H(X),Z)\end{gathered}
            \;\Big|\; (X,Z)\gets\Acal^H
        \Big]\;.
    \end{align*}
    This implies that
    $\SD\lrp{\Dcal^{\Osf_{0}^\Xcal},\Dcal^{\Osf_{1}^\Xcal}} = \epsilon$, and so applying our lifting theorem (\cref{thm:ind_lift}) to $\Dcal$ would give us a similarly efficient distinguisher $\Ecal$ such that
    $$
    \SD\lrp{\Dcal^{\Osf_{0}^\Xcal},\Dcal^{\Osf_{1}^\Xcal}} \leq 8(q+1)^2\cdot\SD\lrp{\Ecal^{\Osf_{0}}, \Ecal^{\Osf_{1}}}\;.
    $$
    We note that since $\Osf_0^\Xcal$ and $\Osf_1^\Xcal$ are efficiently simulatable, they have been internalized as a part of $\Ecal$.

    Finally, let $E$ be the event where $\Ecal^{\Osf_b}$ has made a challenge query $z$ such that $V(Y,z)$ is satisfied. Conditioned on $E$ not happening, the two games behave identically, and moreover the probability of it happening is independent of the choice of $b$. Therefore
    \begin{align*}
        \SD(\Ecal^{\Osf_0},\Ecal^{\Osf_1}) \leq \Pr\sqbkt{E}\;.
    \end{align*}
    On the other hand, when $E$ happens, a $Z$ such that $V(Y,Z)$ is satisfied can indeed be extracted via scanning through the challenge queries made by $\Ecal$. This would then give us a similarly efficient reduction algorithm $\Bcal$ that, on input $Y$, simulates the oracle $\Osf_b$ (the choice of $b$ is irrelevant), and in the end extracts and outputs the aforementioned $Z$ if it exists. The concrete bound obtained here is as follows:
    \begin{align*}
        \Pr\sqbkt{E}\leq \Pr\left[V(Y,Z)\;\middle|\;\begin{gathered}
        Y\gets \Ycal\\
        Z\gets\Bcal(Y)
    \end{gathered}\right]\;.
    \end{align*}
    Collecting the bounds so far, we conclude the proof.
\end{proof}

\end{document}